\documentclass[12pt,draftclsnofoot,onecolumn]{IEEEtran}

\usepackage[T1]{fontenc}
\usepackage{amsfonts}
\usepackage{amsmath}
\usepackage{amsthm}
\usepackage{amssymb}
\usepackage{cite,graphicx,algorithm}
\usepackage{algorithmic}
\usepackage{color}
\usepackage{bm}

\DeclareMathOperator*{\argmax}{arg\,max}

\ifCLASSINFOpdf
\else
\fi
\theoremstyle{plain}
\newtheorem{theorem}{Theorem}

\newtheorem{proposition}{Proposition}

\begin{document}
%
\title{Receiver Design for OTFS with\\ Fractionally Spaced Sampling Approach}
%
%
%

\author{Yao~Ge,~\IEEEmembership{Student Member,~IEEE,}
       Qinwen~Deng,~\IEEEmembership{Student Member,~IEEE,}\\
       P.~C.~Ching,~\IEEEmembership{Fellow,~IEEE,}
       and~Zhi~Ding,~\IEEEmembership{Fellow,~IEEE}
\thanks{Y. Ge and P. C. Ching are with the Department
of Electronic Engineering, The Chinese University of Hong Kong, Hong Kong SAR of China (e-mail: yaoge.gy.jay@hotmail.com; pcching@ee.cuhk.edu.hk).}
\thanks{Q. Deng and Z. Ding are with the Department of Electrical and Computer Engineering, University of California at Davis, Davis, CA 95616 USA (e-mail: mrdeng@ucdavis.edu; zding@ucdavis.edu).}}

%
%

\markboth{}%
{Shell \MakeLowercase{\textit{et al.}}: Bare Demo of IEEEtran.cls for IEEE Communications Society Journals}
%



\maketitle

\begin{abstract}
The recent emergence of orthogonal time frequency space (OTFS) modulation
as a novel PHY-layer mechanism is more suitable in high-mobility wireless communication scenarios than traditional orthogonal frequency division multiplexing (OFDM). Although multiple studies have analyzed OTFS performance using theoretical and ideal
baseband pulseshapes, a challenging and open problem is the
development of effective receivers for practical OTFS systems
that must rely on non-ideal pulseshapes for transmission.
This work focuses on the design of practical receivers for OTFS.
We consider a fractionally spaced sampling (FSS) receiver in which
the sampling rate is an integer multiple of the symbol rate.
For rectangular pulses used in OTFS transmission, we derive a general channel input-output relationship of OTFS in delay-Doppler domain
without the common reliance on impractical assumptions such as ideal
bi-orthogonal pulses and on-the-grid delay/Doppler shifts.
We propose two equalization algorithms: iterative combining message passing (ICMP) and turbo message passing (TMP) for
symbol detection by exploiting delay-Doppler channel sparsity and
the channel diversity gain via FSS.
We analyze the convergence performance of TMP receiver
and propose simplified message passing (MP) receivers to further reduce
complexity. Our FSS receivers demonstrate stronger performance
than traditional receivers
and robustness to the imperfect channel state information knowledge.
\end{abstract}


%
\IEEEpeerreviewmaketitle

\section{Introduction}
%
%
%
%


In widespread development of wireless networks,
high-mobility applications such as high-speed trains and autonomous
vehicles pose new challenges due to the well-known
obstacle of time-varying channels with high Doppler spread.
Even though orthogonal frequency division multiplexing (OFDM) modulation
has achieved high spectral efficiency and throughput
for slow fading frequency selective channels,
its performance degrades significantly against faster
time-varying channels because of the loss of
orthogonality or inter-carrier-interference (ICI)
among OFDM subcarriers. One solution
is to shorten OFDM symbol duration so that the
channel appears quasi-stationary over each OFDM symbol
\cite{wang2006performance}, but at the cost of
lower spectral efficiency caused by the more significant cyclic prefix (CP).
Another approach is to mitigate ICI \cite{cai2003bounding,das2007max,zhao2001intercarrier}, which, however, is only effective for low or medium Doppler shifts and may incur some performance loss.
In addition, \cite{dean2017new} proposes a frequency-domain multiplexing with frequency-domain cyclic prefix (FDM-FDCP) scheme, which can efficiently tackle the Doppler spread but cannot handle the multipath delay effect resulting in inter-symbol interference (ISI).

Recently, orthogonal time frequency space (OTFS) has emerged \cite{hadani2017orthogonal} as a promising PHY-layer
modulation for high-mobility scenarios.
OTFS can exploit the degrees of freedom in both the delay and
Doppler dimensions of a mobile wireless channel,
resulting in superior performance compared with OFDM.
A number of studies on OTFS have been published for
multiple-input multiple-output (MIMO) system \cite{ramachandran2018mimo}, multiple access system \cite{khammammetti2018otfs,ding2019otfs} and
for radar \cite{raviteja2019orthogonal,gaudio2019performance}.
Works by
\cite{raviteja2019otfs} and \cite{surabhi2019diversity}\cite{raviteja2019effective} analyzed the diversity gain of OTFS system in static multipath channels and doubly dispersive channels, respectively. Furthermore,
the authors of \cite{surabhi2019peak} analyzed the
peak-to-average power ratio (PAPR) of OTFS whereas
the authors of
\cite{raviteja2018practical} studied the pulse shaping effect of OTFS.

Unlike OFDM,  OTFS multiplexes information symbols in the 2-dimensional (2D)
delay-Doppler domain instead of the time-frequency domain.
Thus, a channel that rapidly varies in time-frequency domain is
transformed into a near stationary channel in
the delay-Doppler domain.
This near stationary channel simplifies not only  the receiver design \cite{surabhi2019low,surabhi2019diversity,murali2018otfs} but also the
process of channel estimation \cite{raviteja2019embedded,shen2019channel,murali2018otfs} for OTFS systems in high-mobility scenarios.
However, most existing works \cite{murali2018otfs,hadani2017orthogonal,surabhi2019diversity,
ramachandran2018mimo,surabhi2019low,khammammetti2018otfs,ding2019otfs}
only consider the use of the ideal bi-orthogonal pulses that admit
a simple input-output channel relationship
in the delay-Doppler domain for efficient receiver design.
Unfortunately, such ideal pulses are not realizable
in practice due to the Heisenberg uncertainty principle
\cite{matz2013time}.
Alternatively, an OFDM-based OTFS system \cite{farhang2017low,rezazadehreyhani2018analysis,li2019low,long2019low}
may utilize the practical rectangular pulses.
However, this OFDM-based OTFS system will result in low spectral efficiency by inserting a CP in every OFDM symbol of each OTFS frame.

For better spectral efficiency, the works in
\cite{raviteja2018interference} and \cite{tiwari2019low} considered
the use of
rectangular pulses by inserting only one CP for the whole OTFS frame.
To this end,  low-complexity receivers designed in \cite{raviteja2018interference} and \cite{tiwari2019low}
can effectively eliminate self-interference and improve
receiver performance. The assumptions of
\cite{raviteja2018interference} and \cite{tiwari2019low} require
that delay and/or Doppler shifts land on the delay-Doppler
sampling grid which is determined a priori,
however, are still impractical in real OTFS deployment.

In this paper, we investigate more effective receiver algorithms
for OTFS modulation based on rectangular pulses with
a single CP for the entire OTFS frame
as described in\cite{raviteja2018interference}\cite{tiwari2019low}.
We note that existing receiver designs do not
fully utilize the spectral information by
applying restriction to symbol spaced sampling (SSS) for baseband
signal processing.
To preserve sufficient statistic of the OTFS channel output,
we shall apply fractionally spaced sampling (FSS) by sampling at
a rate that is multiple integer of the symbol rate.
Previous results \cite{ding2001blind} have shown that FSS of
signals with sufficient bandwidth can generate a single-input multiple-output (SIMO) channel model and exploit the
underlying channel diversity gain. Our work is motivated
by the fact that OFDM systems under FSS \cite{tepedelenlioglu2004low,wu2010oversampled} have already demonstrated
superior performance over their SSS counterparts.

We propose to use FSS receiver architecture
for OTFS system to achieve high diversity gain under high-mobility time-varying channels in our study. We consider the practical rectangular pulses
and efficiently apply only one CP for each OTFS frame. In addition, we drop the impractical assumption that delay or Doppler shifts
are on the grid and design two efficient receivers
to mitigate ISI in OTFS modulation.
Our contributions in this paper are as follows:
\begin{enumerate}
\item By utilizing the simple and
practical rectangular pulses at the transmitter and receiver,
we derive a general channel input-output relationship for OTFS
in the delay-Doppler domain without relying on the
assumptions such as ideal bi-orthogonal pulses which may not
even exist, or on-the-grid delay/Doppler shifts.
For such practical cases, ISI and extraneous phase
shifts become inevitable at the receiver. We develop novel
effective receiver algorithms to overcome these
practical challenges.

\item We design an OTFS receiver structure based on
FSS and develop two efficient receivers of moderate
complexity for symbol detection. Specifically, we propose an iterative combining
message passing (ICMP) receiver and turbo message passing (TMP)
receiver to exploit the delay-Doppler channel sparsity
and the channel diversity gain via FSS.

\item We analyze the performance and convergence of the proposed TMP receiver by using
extrinsic information transfer (EXIT) chart.
More importantly, we propose a simplified
message passing (MP) algorithm to further reduce the complexity by truncating
weak connection edges in a factor graph
without significant performance loss.
\item Our proposed FSS receivers for OTFS can achieve stronger
performance than the existing solutions.
Both ICMP and TMP receivers exhibit robustness
to uncertainty in channel state information (CSI).
\end{enumerate}

We organize the remainder of this paper as follows: Section
\ref{II_OTFS} introduces the fundamentals of OTFS. Section \ref{III_SSH} characterizes
the channel
input-output relationship of OTFS in the
delay-Doppler domain for non-ideal baseband
pulseshaping. In Section \ref{IV_FSS}, we
first describe the proposed OTFS receiver structure
based on FSS and propose two efficient receivers
for OTFS symbol detection. We further analyze the
performance of the proposed TMP receiver. Section \ref{S_MP}
proposes
a simplified MP algorithm to achieve
good complexity and performance trade-off.
Section \ref{simulation} provides simulation results
of the proposed receivers under the use of
practical baseband pulseshapes.
Finally, Section \ref{Conclusion} concludes our work.
Some detailed proofs appear in the Appendix of the paper.

\section{Fundamentals of OTFS}\label{II_OTFS}
This section briefly outlines basic OTFS concepts and
system model. We present the mathematical description
of conventional OTFS formulation.

\subsection{Basic Concepts of OTFS}
The discrete time-frequency signal plane
consists of time and frequency axes
with respective sampling interval of $T$ (seconds) and
$\Delta f{\rm{ = }}{1 \mathord{\left/
 {\vphantom {1 T}} \right.
 \kern-\nulldelimiterspace} T}$
(Hz), i.e., $$\Lambda  = \left\{ {(nT,m\Delta f),n = 0, \cdots ,N - 1;m = 0, \cdots ,M - 1} \right\}, \; N \in
{\cal Z}, \; M \in
{\cal Z}.
$$
Signals placed on time-frequency grids
denoted by $X[n,m]$, $n = 0, \cdots ,N - 1,m = 0, \cdots ,M - 1$
are transmitted over one OTFS frame with time
duration ${T_f} = NT$ and occupies a bandwidth $B = M\Delta f$.

The corresponding delay-Doppler plane consists
of the message-bearing grids $$\Gamma  = \left\{ {\left(\frac{k}{{NT}},\frac{\ell}{{M\Delta f}}\right),k = 0, \cdots ,N - 1;\ell = 0, \cdots ,M - 1} \right\},$$ where ${1 \mathord{\left/
{\vphantom {1 {M\Delta f}}} \right.
 \kern-\nulldelimiterspace} {M\Delta f}}$ and ${1 \mathord{\left/
 {\vphantom {1 {NT}}} \right.
 \kern-\nulldelimiterspace} {NT}}$ represent the quantization steps of the delay and Doppler frequency, respectively. The choices for $T$ and ${\Delta f}$ are determined by the channel characteristics, i.e., $T$ is not smaller than the maximal delay spread, and ${\Delta f}$ is not smaller than the largest Doppler shift.

At baseband, we can select transmit and receive pulses
${g_{tx}}(t)$ and ${g_{rx}}(t)$, respectively.
Let ${A_{{g_{rx}},{g_{tx}}}}(t,f)$ denotes the cross-ambiguity function between ${g_{tx}}(t)$ and ${g_{rx}}(t)$, i.e.,
\begin{align}
{A_{{g_{rx}},{g_{tx}}}}(t,f) \buildrel \Delta \over = \int {g_{rx}^*(t' - t){g_{tx}}(t'){e^{ - j2\pi f(t' - t)}}dt'}.
\end{align}
In order to fully
eliminate the cross-symbol interference at the receiver, ${g_{tx}}(t)$ and ${g_{rx}}(t)$ should satisfy the following \textbf{bi-orthogonal condition},
\begin{align}\label{Bi_Perfect}
{\left. {{A_{{g_{rx}},{g_{tx}}}}(t,f)} \right|_{t = nT,f = m\Delta f}}=\int {{e^{ - j2\pi m\Delta f(t - nT)}}g_{rx}^*(t - nT){g_{tx}}(t)dt}  = \delta [m]\delta [n].
\end{align}

\subsection{OTFS System Model}
\begin{figure}
  \centering
  \includegraphics[width=6in]{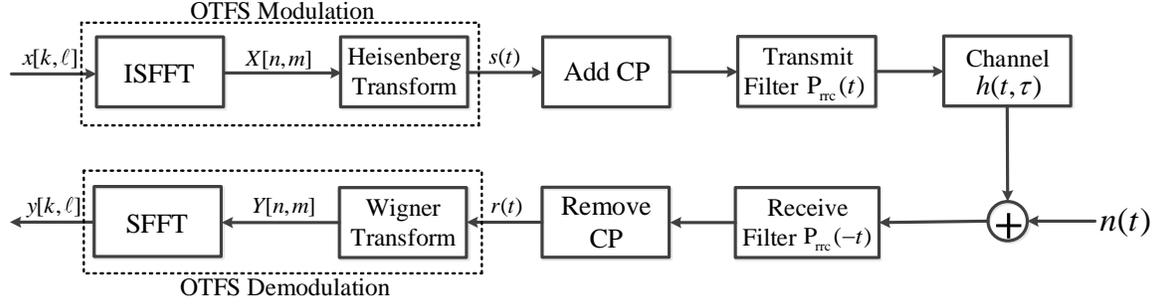}
  \caption{Block diagram of OTFS transmitter (top) and receiver (bottom).
  }\label{OTFS_system}
\end{figure}
The baseband diagram of OTFS system is
given in Fig. \ref{OTFS_system}. Specifically, OTFS modulation
starts with a cascade of a pair of 2D
transforms at the transmitter.
The modulator first maps the information symbols $x[k,\ell]$ in the delay-Doppler domain to $X[n,m]$ in time-frequency plane by
using the inverse symplectic finite Fourier transform (ISFFT).
Consider the $NM$ data
symbols $\left\{ {x[k,\ell],k = 0, \cdots ,N - 1;\ell = 0,
\cdots ,M - 1} \right\}$ from a modulation alphabet $\mathbb{A} = \left\{ {{a_1},{a_2}, \cdots ,{a_Q}} \right\}$ (e.g., QAM
symbols), which are placed on the delay-Doppler plane $\Gamma $.
By using the ISFFT, the $NM$ symbols are converted
into the time-frequency plane
$\Lambda $:
\begin{align}\label{transmit_X}
X[n,m] = \frac{1}{{\sqrt {NM} }}\sum\limits_{k = 0}^{N - 1} {\sum\limits_{\ell = 0}^{M - 1} {x[k,\ell]{e^{j2\pi \left( {\frac{{nk}}{N} - \frac{{m\ell}}{M}} \right)}}} }, \;
n = 0, \cdots ,N - 1;\; m = 0, \cdots, M - 1.
\end{align}
Next, time-frequency signals $X[n,m]$
are transformed into a time domain signal $s(t)$ through
Heisenberg transform utilizing transmit pulse ${g_{tx}}(t)$:
\begin{align}\label{transmit_s}
s(t) = \sum\limits_{n = 0}^{N - 1} {\sum\limits_{m = 0}^{M - 1} {X[n,m]{g_{tx}}(t - nT){e^{j2\pi \left( {m{\rm{ - }}\frac{{M{\rm{ - 1}}}}{{\rm{2}}}} \right)\Delta f(t - nT)}}} }.
\end{align}
We apply a CP of length at least equal to
the maximum baseband channel delay spread.
After inserting a CP in $s(t)$ to tackle inter-frame
interference\footnote{Note that an OFDM-based OTFS system
	proposed in \cite{farhang2017low,rezazadehreyhani2018analysis,li2019low,long2019low} inserts a CP to each of the $N$ OFDM symbols in an
OTFS frame, requiring $N$ CPs per OTFS frame.
Using one CP for each OTFS frame here
can considerably reduce the CP overhead.},
the time domain signal passes through the transmit filter
before entering the (baseband) channel with
baseband impulse response
\begin{align}\label{channel_mod}
h(t,p{T_s}) = \sum\limits_{i = 1}^L {{h_i}{e^{j2\pi {\nu _i}\left( {t - p{T_s}} \right)}}{{\mathop{\rm P}\nolimits} _\text{{rc}}}(p{T_s} - {\tau _i})}, \; p = 0, \cdots ,P - 1,
\end{align}
where $L$ is the number of multipaths and ${T_s} = {1 \mathord{\left/
 {\vphantom {1 {M\Delta f}}} \right.
 \kern-\nulldelimiterspace} {M\Delta f}}$ is the SSS interval;
${{h_i}}$, ${{\tau _i}}$ and ${{\nu _i}}$ represent
the gain, delay and Doppler shift associated with the $i$-th path, respectively.

Note that
${{{\mathop{\rm P}\nolimits} _\text{{rc}}}(p{T_s} - {\tau _i})}$ represents the sampled equivalent
filter response that
includes bandlimiting pulse-shaping
filters used by both transmitter and
receiver to control signal bandwidth and to reject out-of-band
interferences. Generally speaking, ${{{\mathop{\rm P}\nolimits} _\text{{rc}}}(\tau )}$ is a raised-cosine (RC) rolloff
pulse if the transmit filter response is a
root raised-cosine (RRC) rolloff pulse and the receive filter is its corresponding matched filter.
In addition, we denote the Doppler tap for the $i$-th path as ${\nu _i} = ({{k_{{\nu _i}}} + {\beta _{{\nu _i}}}})/NT$,
where integer ${{k_{{\nu _i}}}}$ represents the index of
Doppler frequency ${\nu _i}$ and
${\beta _{{\nu _i}}} \in \left( { -0.5,0.5} \right]$
is the fractional shift of the
nearest Doppler tap ${{k_{{\nu _i}}}}$.
The channel order $P$ is chosen according to the duration of the filter response and the maximum channel delay spread.

At the receiver, the received signal
enters a user-defined receive filter before CP removal.
The received signal $r(t)$ is given by
\begin{align}\label{receive_r}
r(t) = \sum\limits_{p = 0}^{P - 1} {\sum\limits_{i = 1}^L {{h_i}{e^{j2\pi {\nu _i}\left( {t - p{T_s}} \right)}}{{\mathop{\rm P}\nolimits} _\text{{rc}}}(p{T_s} - {\tau _i})} s(t - p{T_s})}  + N(t),
\end{align}
where the filtered noise is $N(t) = \int_\mu  {n(t + \mu ){{\mathop{\rm P}\nolimits} _\text{{rrc}}}(\mu )d\mu } $.
Note that ${{{\mathop{\rm P}\nolimits} _\text{{rrc}}}(\mu )}$
is typically an RRC rolloff receive
filter and ${n(t)}$ represents
the additive white Gaussian noise (AWGN) at the receiver.

The resulting time domain signal $r(t)$ is transformed back to the time-frequency domain through Wigner transform (i.e.,
inverse of Heisenberg transform).
The Wigner transform computes the cross-ambiguity function ${A_{{g_{rx}},r}}(t,f)$ given by
\begin{align}\label{receive_Y}
Y(t,f) = {A_{{g_{rx}},r}}(t,f) \buildrel \Delta \over = \int {g_{rx}^*(t' - t)r(t'){e^{ - j2\pi f(t' - t)}}dt'},
\end{align}
and the SSS baseband received signal
output is obtained by sampling $Y(t,f)$ as
\begin{align}
Y[n,m] = {\left. {Y(t,f)} \right|_{t = nT,f = \left( {m{\rm{ - }}\frac{{M{\rm{ - 1}}}}{{\rm{2}}}} \right)\Delta f}},
\; {n = 0, \cdots ,N - 1};\; {m = 0, \cdots ,M - 1}.
\end{align}

Finally, the symplectic finite Fourier transform (SFFT)
recovers the delay-Doppler domain data symbol
\begin{align}\label{relat_SFFT}
y[k,\ell] = \frac{1}{{\sqrt {NM} }}\sum\limits_{n = 0}^{N - 1} {\sum\limits_{m = 0}^{M - 1} {Y[n,m]{e^{ - j2\pi \left( {\frac{{nk}}{N} - \frac{{m\ell}}{M}} \right)}}} },
\; {k = 0, \cdots ,N - 1;\;\ell = 0, \cdots ,M - 1}.
\end{align}
These operations provide the basis of
OTFS model with SSS approach in a general case.
They are very useful to further study OTFS
system when the specific pulses are employed.

For analytical convenience,  we capture the relationship between
$X[n,m]$ and output $Y[n,m]$
in the following theorem.
\begin{theorem}\label{relation_TF_Theorem}
	The input-output relationship of OTFS in time-frequency domain is given by
	\begin{align}\label{relation_TF_G}
	Y[n,m] = \sum\limits_{n' = 0}^{N - 1} {\sum\limits_{m' = 0}^{M - 1} {{H_{n,m}}[n',m']X[n',m']} }  + V[n,m],
	\end{align}
	where $V[n,m]$ is the noise at the output of the Wigner transform and
	\begin{align}\label{TF_Channel}
	{H_{n,m}}[n',m'] &= \sum\limits_{p = 0}^{P - 1} {\sum\limits_{i = 1}^L {{h_i}{{\mathop{\rm P}\nolimits} _\text{{rc}}}(p{T_s} - {\tau _i}){A_{{g_{rx}},{g_{tx}}}}\left( {(n - n')T - p{T_s},(m - m')\Delta f - {\nu _i}} \right)} }\nonumber\\
	&\quad \times {e^{j\pi \left( {M - 1} \right)\Delta f(p{T_s} - (n - n')T)}}{e^{j2\pi m'\Delta f((n - n')T - p{T_s})}}{e^{j2\pi {\nu _i}\left( {nT - p{T_s}} \right)}}.
	\end{align}
\end{theorem}
\begin{proof}
	See Appendix \ref{Appe_A}.
\end{proof}
We can also characterize the relationship
between channel output $y[k,\ell]$ and input $x[k,\ell]$
in the following theorem.
\begin{theorem}\label{relation_DD_Theorem}
The input-output relationship of OTFS in delay-Doppler domain is given by
\begin{align}\label{relation_DD_G}
y[k,\ell] = \frac{1}{{NM}}\sum\limits_{k' = 0}^{N - 1} {\sum\limits_{\ell' = 0}^{M - 1} {{h_{k,\ell}}[k',\ell']x[k',\ell']} }  + \upsilon [k,\ell],
\end{align}
where $\upsilon [k,\ell] = \frac{1}{{\sqrt {NM} }}\sum\limits_{n = 0}^{N - 1} {\sum\limits_{m = 0}^{M - 1} {V[n,m]{e^{ - j2\pi \left( {\frac{{nk}}{N} - \frac{{m\ell}}{M}} \right)}}} }$ and
\begin{align}\label{rela_h}
{h_{k,\ell}}[k',\ell'] = \sum\limits_{n = 0}^{N - 1} {\sum\limits_{m = 0}^{M - 1} {\sum\limits_{n' = 0}^{N - 1} {\sum\limits_{m' = 0}^{M - 1} {{H_{n,m}}[n',m']} } {e^{ - j2\pi \left( {\frac{{nk}}{N} - \frac{{m\ell}}{M}} \right)}}{e^{j2\pi \left( {\frac{{n'k'}}{N} - \frac{{m'\ell'}}{M}} \right)}}} }.
\end{align}
\end{theorem}
\begin{proof}
See Appendix \ref{Appe_B}.
\end{proof}

In the next section, we will consider a practical communication system, where the rectangular pulses are adopted by both the transmitter and receiver.

\section{OTFS Model Based on SSS for Rectangular Pulses}\label{III_SSH}

Recall that many existing works on OTFS \cite{hadani2017orthogonal,ramachandran2018mimo,khammammetti2018otfs,ding2019otfs,
	raviteja2019orthogonal,gaudio2019performance,raviteja2019otfs,surabhi2019diversity,
	raviteja2019effective,raviteja2018practical,surabhi2019low,murali2018otfs,
	raviteja2019embedded,raviteja2018interference,tiwari2019low} relied on certain impractical assumptions such as ideal
bi-orthogonal pulses and on-the-grid
delay and/or Doppler shifts.
In this section, we first prove that OTFS with rectangular pulses
at both transmitter and receiver, the bi-orthogonal condition in (\ref{Bi_Perfect}), can be satisfied. However, for time-varying
channels,
the ideal bi-orthogonal condition in (\ref{Bi_ideal})
does not hold.
We then derive a general input-output
relationship of OTFS system in delay-Doppler domain
for SSS.

Let ${\bar u}(t)$ denote the unit step function.
Without loss of generality, we consider rectangular pulses
$\mbox{rect}(t)= T^{-1/2} \cdot [{\bar u}(t)-{\bar u}(t-T)]$.
Given rectangular transmitter and receiver OTFS pulses
${g_{tx}}(t)=g_{rx}(t)=\mbox{rect}(t)$,
we have the following result:
\begin{proposition}
The rectangular pulses used by both the transmitter and receiver can satisfy the bi-orthogonal condition in (\ref{Bi_Perfect}), i.e.,
\begin{align}
\int {{e^{ - j2\pi m\Delta f(t - nT)}}g_{rx}^*(t - nT){g_{tx}}(t)dt}  = \delta [m]\delta [n].
\end{align}
\end{proposition}
\begin{proof}
For $n \ne 0$, we clearly have $\int {{e^{ - j2\pi m\Delta f(t - nT)}}g_{rx}^*(t - nT){g_{tx}}(t)dt}  = 0$ due to the
finite time duration $T$ of the rectangular pulses $\mbox{rect}(t)$.
For $n=0$, we have
\begin{align*}
\int {{e^{ - j2\pi m\Delta f(t - nT)}}g_{rx}^*(t - nT){g_{tx}}(t)dt} = \frac{1}{T}\int_0^T {{e^{ - j2\pi m\Delta ft}}dt} = \delta[m].
\end{align*}
The proof is complete by combining both cases.
\end{proof}

However, when incorporating time-varying channel
(\ref{channel_mod}),  rectangular pulses cannot
guarantee the following \textbf{ideal bi-orthogonal condition}
\begin{align}\label{Bi_ideal}
A_{g_{rx},g_{tx}}(t,f)& = \delta [m]\delta[n]
q_{(-(P - 1){T_s},0)}(t)
q_{( - \nu _{\max }, \nu _{\max })} (f),  \begin{array}{l}
 t = \left( nT-(P - 1){T_s},nT \right),
\\f = (-\nu_{\max} + m\Delta f, \nu _{\max }+ m\Delta f),
\end{array}
\end{align}
where ${q_{(a,b)}}(x) = 1$ for $x \in (a,b)$ and $0$ otherwise.
This \textbf{ideal bi-orthogonal condition} in (\ref{Bi_ideal})
ensures that the ISI is eliminated at the
receiver in a practical communication system. However, an ideal pulses which satisfy the above \textbf{ideal bi-orthogonal
	condition} cannot be realized in practice due to
Heisenberg uncertainty principle \cite{matz2013time}. Thus,
the ISI is inevitable
at receiver input such that receiver equalization
is necessary for satisfactory reception performance.

Considering the rectangular pulses and the CP effect, we can rewrite ${{H_{n,m}}[n',m']}$ in (\ref{TF_Channel}) as
\begin{align}\label{relat_H_R_CP}
{H_{n,m}}[n',m'] &= \sum\limits_{p = 0}^{P - 1} {\sum\limits_{i = 1}^L {{h_i}{{\mathop{\rm P}\nolimits} _\text{{rc}}}(p{T_s} - {\tau _i}){A_{{g_{rx}},{g_{tx}}}}\left( {{{\left[ {n - n'} \right]}_N}T - p{T_s},(m - m')\Delta f - {\nu _i}} \right)} } \nonumber\\
&\quad \times {e^{j\pi \left( {M - 1} \right)\Delta f\left( {p{T_s} - {{\left[ {n - n'} \right]}_N}T} \right)}}{e^{j2\pi m'\Delta f\left( {{{\left[ {n - n'} \right]}_N}T - p{T_s}} \right)}}{e^{j2\pi {\nu _i}\left( {nT - p{T_s}} \right)}},
\end{align}
where the cross-ambiguity function
${{A_{{g_{rx}},{g_{tx}}}}\left( {{{\left[ {n - n'} \right]}_N}T - p{T_s},(m - m')\Delta f - {\nu _i}} \right)}$ is non-zero for $p = 0, \cdots ,P - 1$ and $\left| {{\nu _i}} \right| < {\nu _{\max }}$ only when ${\left[ {n - n'} \right]_N} \le 1$,
 i.e., $n' = n$ and $n' = {\left[ {n - 1} \right]_N}$. Hence, the time-frequency relationship
in (\ref{relation_TF_G}) reduces to
\begin{align}
Y[n,m] &= {H_{n,m}}[n,m]X[n,m] + \sum\limits_{m' = 0,m' \ne m}^{M - 1} {{H_{n,m}}[n,m']X[n,m']}  \nonumber\\
& \quad + \sum\limits_{m' = 0}^{M - 1} {{H_{n,m}}\left[ {{{\left[ {n - 1} \right]}_N},m'} \right]X\left[ {{{\left[ {n - 1} \right]}_N},m'} \right]}  + V[n,m],
\end{align}
in which the first term contains the desired signal,
the second and third terms represent ICI and ISI,
respectively. The following theorem summarizes the findings:
\begin{theorem}\label{relation_DDR_Theorem}
The OTFS input-output relationship in delay-Doppler domain
with rectangular pulses is given by
\begin{align}\label{relation_DD_R}
y[k,\ell] = \sum\limits_{p = 0}^{P - 1} {\sum\limits_{i = 1}^L {\sum\limits_{q = 0}^{N - 1} {{h_i}{{\mathop{\rm P}\nolimits} _\text{{rc}}}(p{T_s} - {\tau _i})\gamma (k,\ell,p,q,{k_{{\nu _i}}},{\beta _{{\nu _i}}})x\left[ {{{\left[ {k - {k_{{\nu _i}}} + q} \right]}_N},{{\left[ {\ell - p} \right]}_M}} \right]} } }  + \upsilon [k,\ell],
\end{align}
where
\begin{subequations}
\begin{equation}\label{gamma_off}
\gamma (k,\ell,p,q,{k_{{\nu _i}}},{\beta _{{\nu _i}}}) =
\begin{cases}
\frac{1}{N}\xi (\ell,p,{k_{{\nu _i}}},{\beta _{{\nu _i}}})\theta (q,{\beta _{{\nu _i}}}),&p \le \ell < M,\\
\frac{1}{N}\xi (\ell,p,{k_{{\nu _i}}},{\beta _{{\nu _i}}})\theta (q,{\beta _{{\nu _i}}})\phi (k,q,{k_{{\nu _i}}}), &0 \le \ell < p,
\end{cases}
\end{equation}
\begin{align}\label{Phase_off}
\xi (\ell,p,{k_{{\nu _i}}},{\beta _{{\nu _i}}}) = {e^{j\pi \frac{{M - 1}}{M}p}}{e^{j2\pi \left( {\frac{{\ell - p}}{M}} \right)\left( {\frac{{{k_{{\nu _i}}} + {\beta _{{\nu _i}}}}}{N}} \right)}},
\end{align}
\begin{align}\label{Theat_off}
\theta (q,{\beta _{{\nu _i}}}) = \frac{{{e^{ - j2\pi ( - q - {\beta _{{\nu _i}}})}} - 1}}{{{e^{ - j\frac{{2\pi }}{N}( - q - {\beta _{{\nu _i}}})}} - 1}},
\end{align}
\begin{align}\label{Phi_off}
\phi (k,q,{k_{{\nu _i}}}) = {e^{ - j\pi \left( {M - 1} \right)}}{e^{ - j2\pi \frac{{{{\left[ {k - {k_{{\nu _i}}} + q} \right]}_N}}}{N}}}.
\end{align}
\end{subequations}
\end{theorem}
\begin{proof}
See Appendix \ref{Appe_C}.
\end{proof}

Note that the magnitude of $\theta (q,{\beta _{{\nu _i}}})$ in (\ref{Theat_off})
peaks at $q = 0$ and decreases rapidly as $|q|$ grows.
Hence, we can only consider a small number $2{E_i} + 1$ (${E_i} \ge 0$) of significant values $\theta (q,{\beta _{{\nu _i}}})$ in (\ref{Theat_off}), i.e., $ - {E_i} \le q \le {E_i}$. By using this approximation, we can
conveniently rewrite the received signal $y[k,\ell]$ in (\ref{relation_DD_R}):
\begin{align}
y[k,\ell] \approx \sum\limits_{p = 0}^{P - 1} {\sum\limits_{i = 1}^L {\sum\limits_{q =  - {E_i}}^{{E_i}} {{h_i}{{\mathop{\rm P}\nolimits} _\text{{rc}}}(p{T_s} - {\tau _i})\gamma (k,\ell,p,q,{k_{{\nu _i}}},{\beta _{{\nu _i}}})x\left[ {{{\left[ {k - {k_{{\nu _i}}} + q} \right]}_N},{{\left[ {\ell - p} \right]}_M}} \right]} } }  + \upsilon [k,\ell].
\end{align}
In addition, the relationship of (\ref{relation_DD_R}) can be simplified as follows if Doppler shifts are exactly on the grid
such that ${\beta _{{\nu _i}}}{\rm{ = 0,}}\;\forall i$ without
fractional Doppler shift:

\begin{proposition}\label{relation_DDRS_Theorem}
For Doppler shifts exactly on the grid, the relationship of
(\ref{relation_DD_R}) reduces to
\begin{align}\label{relation_DD_R_on}
y[k,\ell] = \sum\limits_{p = 0}^{P - 1} {\sum\limits_{i = 1}^L {{h_i}{{\mathop{\rm P}\nolimits} _{rc}}(p{T_s} - {\tau _i})\gamma (k,\ell,p,{k_{{\nu _i}}})x\left[ {{{\left[ {k - {k_{{\nu _i}}}} \right]}_N},{{\left[ {\ell - p} \right]}_M}} \right]} }  + \upsilon [k,\ell],
\end{align}
where
\begin{equation}
\gamma (k,\ell,p,{k_{{\nu _i}}}) =
\begin{cases}
\xi (\ell,p,{k_{{\nu _i}}},0),&p \le \ell < M,\\
\xi (\ell,p,{k_{{\nu _i}}},0){e^{ - j\pi \left( {M - 1} \right)}}{e^{ - j2\pi \frac{{{{\left[ {k - {k_{{\nu _i}}}} \right]}_N}}}{N}}}, &0 \le \ell < p.
\end{cases}
\end{equation}
\end{proposition}
\begin{proof}
The proof follows directly by noting from (\ref{Theat_off})
that
\begin{equation}\label{Theat_on}
\theta (q,0) = \sum_{n=0}^{N-1} e^{j\frac{2\pi}{N} n q } =
\begin{cases}
N,& [q]_N=0,\\
0, &\text{otherwise}.
\end{cases}=N\delta\left[[q]_N\right].
\end{equation}
Here we defined $[q]_N$ as the remainder of $q$ dividing $N$.
Accordingly, the result in (\ref{relation_DD_R_on}) follows from (\ref{relation_DD_R}).
\end{proof}

From \textbf{Theorem \ref{relation_DDR_Theorem}}, we
observe that the ISI and extra phase shifts at the receiver can
affect the symbol detection when OTFS uses the
practical rectangular pulses in Heisenberg and Wigner
transforms. Even the simplified model for
 on-the-grid Doppler shifts in \textbf{Proposition \ref{relation_DDRS_Theorem}}, the
ISI is still present.
Therefore, simple and effective receiver must
be designed to recover the signal in such practical
and non-ideal OTFS setups.

\section{Receiver Design for OTFS with FSS}\label{IV_FSS}

Recalling from the literature such as \cite{forney1972maximum,ding2001blind} that
the use of RRC filter at the transmitter
and the matched receive filter would widen the bandwidth beyond the minimum required bandwidth of $1/2T_s$. Thus, symbol spaced sampling (SSS),
typically, would not preserve
sufficient statistic for signal recovery since it falls below Nyquist
sampling rate and could also lead to performance sensitivity at the sampling instants \cite{forney1972maximum,ding2001blind}.

To develop more effective and robust receiver algorithms,
we attempted to design a fractionally spaced sampling (FSS)
receiver for OTFS with rectangular pulses and RRC filters, which is expected to be able to  generate weakly-correlated noises, admit sufficient statistic
\cite{forney1972maximum,ding2001blind} and further improve the equalization performance. Note that our proposed receivers
can be generalized to the non-rectangular pulses in a straightforward manner following the
steps from Appendices A, B and C.

\subsection{Receiver Structure}

When sampling at a rate that is an integer multiple $G$ of
the symbol rate,
FSS receiver is equivalent to a SIMO linear system in which $G$ multiple parallel channels
have correlated noise \cite{ding2001blind}.
The SIMO channel responses depend on the time-varying
channel as well as transmit and receive filters.
For the receive filter output to be sampled at rate
$G/{{T_s}}$, we first  write its polyphase representation
after the removal of the CP as
\begin{align}
{r_g}[u] = \sum\limits_{p = 0}^{P - 1} {\sum\limits_{i = 1}^L {{h_i}{e^{j2\pi {\nu _i}\left( {uT_s - p{T_s}} \right)}}{{\mathop{\rm P}\nolimits} _\text{{rc}}}[g]} s(uT_s - p{T_s})}
+ {N_g}[u], \begin{array}{l}
 u = 0, \cdots, NM - 1,
\\g = 0, \cdots, G - 1,
\end{array}
\end{align}
where the $g$-th channel output sequence is ${r_g}[u]\stackrel{\triangle}{=} r(uT_s + {{g{T_s}} \mathord{\left/
 {\vphantom {{g{T_s}} G}} \right.
 \kern-\nulldelimiterspace} G})$ with additive noise
${N_g}[u] \stackrel{\triangle}{=}N(uT_s + {{g{T_s}} \mathord{\left/
{\vphantom {{g{T_s}} G}} \right.
 \kern-\nulldelimiterspace} G})$, and ${{\rm{P}}_\text{{rc}}}[g] \buildrel \Delta \over = {{\rm{P}}_\text{{rc}}}(p{T_s} + {{g{T_s}} \mathord{\left/
 {\vphantom {{g{T_s}} G}} \right.
 \kern-\nulldelimiterspace} G} - {\tau _i})$ for simplicity.
The resulting SIMO receiver structure diagram for OTFS system is given in Fig. \ref{Receiver_structure}.
\begin{figure}
  \centering
  \includegraphics[width=6in]{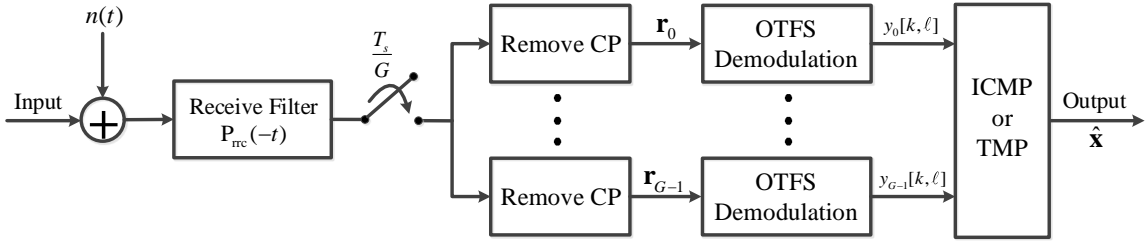}
  \caption{Receiver structure of FSS approach for OTFS system.}\label{Receiver_structure}
\end{figure}

To utilize the multiple receptions for diversity combining, we proceed with our FSS-OTFS system model. By performing an OTFS demodulation at the receiver for each $g$, we obtain the following relationship:
\begin{subequations}\label{FS_relation_DD}
\begin{align}
{y_g}[k,\ell]& = \sum\limits_{p = 0}^{P - 1} {\sum\limits_{i = 1}^L {\sum\limits_{q = 0}^{N - 1} {{h_i}{{\mathop{\rm P}\nolimits} _\text{{rc}}}[g]\gamma (k,\ell,p,q,{k_{{\nu _i}}},{\beta _{{\nu _i}}})x\left[ {{{\left[ {k - {k_{{\nu _i}}} + q} \right]}_N},{{\left[ {\ell - p} \right]}_M}} \right]} } }  + {\upsilon _g}[k,\ell]\\
 & \approx \sum\limits_{p = 0}^{P - 1} {\sum\limits_{i = 1}^L {\sum\limits_{q =  - {E_i}}^{{E_i}} {{h_i}{{\mathop{\rm P}\nolimits} _\text{{rc}}}[g]\gamma (k,\ell,p,q,{k_{{\nu _i}}},{\beta _{{\nu _i}}})x\left[ {{{\left[ {k - {k_{{\nu _i}}} + q} \right]}_N},{{\left[ {\ell - p} \right]}_M}} \right]} } }  + {\upsilon _g}[k,\ell],
\end{align}
\end{subequations}
where ${\upsilon _g}[k,\ell] = \frac{1}{{\sqrt {NM} }}\sum\limits_{n = 0}^{N - 1} {\sum\limits_{m = 0}^{M - 1} {{V_g}[n,m]} }{e^{ - j2\pi \left( {\frac{{nk}}{N} - \frac{{m\ell}}{M}} \right)}}$ with ${V_g}[n,m] = \int g_{rx}^*(t' - nT){N_g}(t')$ $\times {e^{ - j2\pi \left( {m{\rm{ - }}\frac{{M{\rm{ - 1}}}}{{\rm{2}}}} \right)\Delta f(t' - nT)}}dt'$.

The input-output relationship in (\ref{FS_relation_DD}) can be vectorized as
\begin{align}\label{FS_V}
{{\bf{y}}_g}  \simeq  {{\bf{H}}_g}{\bf{x}} + {{\bf{z}}_g},\;g = 0, \cdots ,G - 1,
\end{align}
where ${\bf{x}},{{\bf{y}}_g},{{\bf{z}}_g} \in {\mathbb{C}^{NM \times 1}}$ and ${{\bf{H}}_g} \in {\mathbb{C}^{NM \times NM}}$.
Because of the modulo-$N$ and modulo-$M$ operations in (\ref{FS_relation_DD}), the number of
non-zero elements in each row and column of ${{\bf{H}}_g}$
is identically $D$.
Typically, $D$ is much smaller than $NM$,
leading to a sparse matrix ${{\bf{H}}_g}$.

For the special case of SSS, the receiver is simplified with $G=1$.
One can derive an efficient MP algorithm
for symbol detection by rounding the delay and/or Doppler
shifts to integers on receiver sampling grid \cite{ramachandran2018mimo,raviteja2018interference}.

Because the RRC transmit and receive
bandlimiting filter has bandwidth between $1/2T_s$ and $1/T_s$.
Therefore, selecting $G=2$ as FSS
interval suffices to preserve the sufficient signal statistic.
Thus, we shall focus on the use of $G=2$ henceforth.
Unlike in \cite{ramachandran2018mimo} and \cite{raviteja2018interference}, our proposed
receivers can potentially exploit channel spectrum diversity gain
to substantially improve the performance without
relying on multiple antennas and multiple radio frequency (RF) chains.
We also drop the impractical assumption that delay/Doppler shifts
must be on the grid.
Note that extension to larger $G$, when broader bandwidth
becomes available, is straightforward since we can always
separate $G$ channels into two groups for receiver equalization.

\subsection{ICMP Receiver Equalization}
In this part, we will introduce the ICMP receiver
to take advantage of the SIMO
receptions. To this end, we combine the equations from the receptions in (\ref{FS_V}) as
\begin{align}\label{relation_vec}
{\bf{y}} = {\bf{Hx}} + {\bf{z}},
\end{align}
where ${\bf{y}} = {[{\bf{y}}_0^T,{\bf{y}}_1^T]^T}$, ${\bf{H}} = {[{\bf{H}}_0^T,{\bf{H}}_1^T]^T}$ and ${\bf{z}} = {[{\bf{z}}_0^T,{\bf{z}}_1^T]^T}$.
The basic problem now is to detect the transmitted symbol vector ${\bf{x}}$ from the received signals with channel knowledge at the receiver. Direct solution of (\ref{relation_vec})
could be computationally demanding as it involves
inversion of a large matrix as $NM$ can typically be in the
order of thousands.

Let $\mathcal{I}(d)$ and $\mathcal{J}(c)$
denote the sets of indexes with non-zero elements
in the $d$-th row and $c$-th column of ${\bf{H}}$,
where $d = 1, \cdots ,2NM$ and $c = 1, \cdots ,NM$,
respectively. Here, we can interpret the system model in (\ref{relation_vec}) as a sparsely-connected factor graph.
Thus, we can apply the concept of low-complexity MP algorithm for symbol detection.
Specifically, each entry of the observation vector ${\bf{y}}$
denotes an observation node, whereas each transmitted symbol
is viewed as a variable node.
In this factor graph,
each observation node $y[d]$ is connected to the set of $D$
variable nodes $\{ x[c],c \in \mathcal{I}(d)\} $
whereas each variable node $x[c]$ is connected to the set of $2D$ observation nodes $\{ y[d],d \in \mathcal{J}(c)\} $.

The optimal way of detecting the transmitted symbols is joint maximum a posterior probability (MAP) detection, i.e.,
	\begin{equation*}
	\mathbf{\hat{x}}=\argmax_{\mathbf{x}\in \mathbb{A}^{NM\times 1}}\Pr  \bigl(\mathbf{x|y,H}\bigr),
	\end{equation*}
which has a complexity exponential in $NM$. This can be intractable when the product $NM$ is in the order of
several thousands.
As a result, we derive a suboptimal symbol-by-symbol MAP detection
through following approximation:
	\begin{align}
	\hat{x}[c]&=\argmax_{a_j\in\mathbb{A}}\Pr \bigl(x[c]=a_j|\mathbf{y,H}\bigr)\notag\\
	&=\argmax_{a_j\in \mathbb{A}}\Pr (x[c] = {a_j})\Pr \bigl(\mathbf{y}|x[c]=a_j,\mathbf{H}\bigr)\notag\\
	&\approx \argmax_{a_j\in \mathbb{A}}{\omega _c}({a_j})\prod_{d\in \mathcal{J}(c)}{\Pr(y[d]\left|{x[c] = {a_j},{\bf{H}}} \right.)}\label{independ_y}
	\end{align}
In \eqref{independ_y}, we denote a priori probability
when $x[c]=a_j$ as ${\omega _c}({a_j})$ and assume the
components of $\mathbf{y}$ are approximately independent for
a given $x[c]$ due to the sparsity of $\mathbf{H}$. To further
reduce
complexity, we employ the Gaussian approximation of the
interference components within the proposed ICMP receiver.
Similar to \cite{ramachandran2018mimo,raviteja2018interference}, this algorithm performs
message passing iteratively among observation nodes
and variable nodes according to the factor graph.
The ICMP receiver is summarized in \textbf{Algorithm \ref{alg:A}}.
Below are its detailed steps in iteration $\kappa $:

\begin{algorithm}
\caption{ICMP Receiver}
\label{alg:A}
\begin{algorithmic}
\STATE {Input: ${\bf{y}}$, ${\bf{H}}$, ${\omega _c}({a_j}) = {1 \mathord{\left/
 {\vphantom {1 Q}} \right.
 \kern-\nulldelimiterspace} Q}$, $c = 1, \cdots NM$, $j = 1, \cdots Q$ and $n_{iter}$.}
\STATE {Initialization: ${\bf{p}}_{c,d}^{(0)} = {{\bm{\omega }}_c}$, $c = 1, \cdots NM$, $d\in \mathcal{J}(c)$, $\eta^{(0)}=0$ and iteration count $\kappa=1$.}
\REPEAT
\STATE \begin{enumerate}
         \item Each observation node $y[d]$ computes the mean $\mu _{d,c}^{(\kappa)}$ and variance ${{(\sigma _{d,c}^{(\kappa)})}^2}$ in (\ref{5}) and (\ref{6}), then passes them to the connected variable nodes $x[c],c\in\mathcal{I}(d)$;
         \item Each variable node $x[c]$ generates ${\bf{p}}_{c,d}^{(\kappa)}$ in (\ref{7}) and passes them to the connected observation nodes $y[d],d\in \mathcal{J}(c)$;
         \item Compute the convergence indicator $\eta^{(\kappa)}$ and symbol probabilities ${\bf{p}}_c^{(\kappa)}$ in (\ref{9});
         \item Update decision symbol probabilities ${\bf{\bar p}}_c = {\bf{p}}_c^{(\kappa)}$ if $\eta^{(\kappa)}>\eta^{(\kappa-1)}$;
         \item $\kappa: = \kappa + 1$;
       \end{enumerate}
\UNTIL{$\eta^{(\kappa)}=1$ or $\kappa=n_{iter}$.}
\STATE {Output: The decisions of the transmitted symbols in (\ref{12}).}
\end{algorithmic}
\end{algorithm}

\textbf{From observation node $y[d]$ to variable nodes $x[c],c\in\mathcal{I}(d)$}: At each observation node, extrinsic messages to each connected variable node is computed according
to the channel model,
noisy channel observations,
and a priori information from other connected variable nodes.
The received signal $y[d]$ can be written as
	\begin{equation}
	y[d] = H[d,c]x[c] + \underbrace {\sum\limits_{e \in {\cal I}(d),e \ne c} H[d,e]x[e] + z[d]}_{\zeta _{d,c}^{(\kappa)}},\label{4}
	\end{equation}
where sum of interference and noise $\zeta_{d,c}^{(\kappa)}$ is approximately modeled as $\mathcal{CN}\left( {\mu _{d,c}^{(\kappa)},{{(\sigma _{d,c}^{(\kappa)})}^2}} \right)$
according to Central Limit Theorem \cite{brosamler1988almost},  with
	\begin{equation}
	\mu_{d,c}^{(\kappa)}=\sum_{e\in\mathcal{I}(d),e\neq c}\sum_{j=1}^Q p_{e,d}^{(\kappa-1)}(a_j)a_jH[d,e],\label{5}
	\end{equation}
	\begin{equation}
	{(\sigma _{d,c}^{(\kappa)})^2} = \sum\limits_{e \in {\cal I}(d),e \ne c} {\left( {\sum\limits_{j = 1}^Q {p_{e,d}^{(\kappa - 1)}} ({a_j}){{\left| {{a_j}} \right|}^{\rm{2}}}{{\left| {H[d,e]} \right|}^{\rm{2}}} - {{\left| {\sum\limits_{j = 1}^Q {p_{e,d}^{(\kappa - 1)}} ({a_j}){a_j}H[d,e]} \right|}^{\rm{2}}}} \right)}  + \sigma _N^2.\label{6}
	\end{equation}
In \eqref{6}, $\sigma _N^2 = \sigma _n^2\int_\mu  {{\mathop{\rm P}\nolimits} _\text{{rrc}}^2(\mu )d\mu }$ is the variance of the colored Gaussian noise after the receive filter and $\sigma_n^2$ is the variance of the AWGN $\bf{n}$ at the receiver input.
The mean $\mu _{d,c}^{(\kappa)}$ and variance ${{(\sigma _{d,c}^{(\kappa)})}^2}$ are used as messages passed from observation nodes to variable nodes.
	
\textbf{From variable node $x[c]$ to observation
nodes $y[d],d\in \mathcal{J}(c)$}: At each variable node,
the extrinsic information for each connected observation node is generated from prior messages collected from other observation
nodes. A posteriori log-likelihood ratio (LLR) is given by
\begin{align}
\alpha _{c}^{(\kappa)}({a_j})& = \log \frac{{{\omega _c}({a_j})\prod\limits_{e \in {\cal J}(c)} {\Pr (y[e]\left| {x[c] = {a_j},{\bf{H}}} \right.)} }}{{{\omega _c}({a_Q})\prod\limits_{e \in {\cal J}(c)} {\Pr (y[e]\left| {x[c] = {a_Q},{\bf{H}}} \right.)} }}\nonumber\\
& = \underbrace {\log \frac{{{\omega _c}({a_j})\prod\limits_{e \in {\cal J}(c),e \ne d} {\Pr (y[e]\left| {x[c] = {a_j},{\bf{H}}} \right.)} }}{{{\omega _c}({a_Q})\prod\limits_{e \in {\cal J}(c),e \ne d} {\Pr (y[e]\left| {x[c] = {a_Q},{\bf{H}}} \right.)} }}}_{\alpha _{c,d}^{(\kappa)}({a_j})} + \underbrace {\log \frac{{\Pr (y[d]\left| {x[c] = {a_j},{\bf{H}}} \right.)}}{{\Pr (y[d]\left| {x[c] = {a_Q},{\bf{H}}} \right.)}}}_{\Lambda _{c,d}^{(\kappa)}({a_j})},\nonumber
\end{align}
where $\Lambda _{c,d}^{(\kappa)}({a_j}) = \log \frac{{\varepsilon _{d,c}^{(\kappa)}({a_j})}}{{\varepsilon _{d,c}^{(\kappa)}({a_Q})}}$, the extrinsic LLR $\alpha _{c,d}^{(\kappa)}({a_j}) = \log \frac{{{\omega _c}({a_j})}}{{{\omega _c}({a_Q})}} + \sum\limits_{e \in {\cal J}(c),e \ne d} {\log \frac{{\varepsilon _{e,c}^{(\kappa)}({a_j})}}{{\varepsilon _{e,c}^{(\kappa)}({a_Q})}}}$ and $\varepsilon _{e,c}^{(\kappa)}({a_j}) = \exp \left( { - \frac{{{{\left| {y[e] - \mu _{e,c}^{(\kappa)} - {H_{e,c}}{a_j}} \right|}^2}}}{{{{(\sigma _{e,c}^{(\kappa)})}^2}}}} \right)$.
The message passed from a variable node $x[c]$ to observation nodes $y[d],d\in \mathcal{J}(c)$ is the probability mass function of the alphabet
	\begin{equation}
	p_{c,d}^{(\kappa)}(a_j)=\Delta \cdot \tilde{P}_{c,d}^{(\kappa)}(a_j)+(1-\Delta)\cdot p_{c,d}^{(\kappa-1)}(a_j),\ a_j\in \mathbb{A}, \label{7}
	\end{equation}	
where $\tilde P_{c,d}^{(\kappa)}({a_j}){\rm{ }} = \left[\sum\limits_{k = 1}^Q {\exp \left( {\alpha _{c,d}^{(\kappa)}({a_k})} \right)} \right]^{-1}
{\exp \left( {\alpha _{c,d}^{(\kappa)}({a_j})} \right)}
$ and $\Delta\in(0,1]$ is a \textit{message damping factor} used
to improve performance by controlling convergence speed \cite{som2011low,ramachandran2018mimo,raviteja2018interference}.
	
\textbf{Convergence indicator}: The convergence indicator $\eta^{(\kappa)}$ can be computed as
	\begin{equation}
	\eta^{(\kappa)}=\frac{1}{NM}\sum_{c=1}^{NM}\mathbb{I}\Biggl(\max_{a_j\in\mathbb{A}}\ p_c^{(\kappa)}(a_j)\geq 1-\varrho\Biggr) \label{9}
	\end{equation}
	for some small $\varrho >0$ and where $p_c^{(\kappa)}({a_j}) =\left[{\sum\limits_{k = 1}^Q {\exp \left( {\alpha _c^{(\kappa)}({a_k})} \right)} }\right]^{-1}
{{\exp \left( {\alpha _c^{(\kappa)}({a_j})} \right)}}$. $\mathbb{I}(\cdot)$ denotes the indicator function.
	
\textbf{Update criteria}: If $\eta^{(\kappa)}>\eta^{(\kappa-1)}$, then we update the probabilities of transmitted symbols as
	\begin{equation}
	{\bf{\bar p}}_c = {\bf{p}}_c^{(\kappa)},\ c=1,\dots,NM. \label{11}
	\end{equation}
Note that we only update the probabilities if the current iteration
is better than the previous one.

\textbf{Stopping criteria}: The MP algorithm stops if either $\eta^{(\kappa)}=1$ or the maximum number of iterations $n_{iter}$ is reached.

Once the stopping criteria is satisfied, we make the decisions of the transmitted symbols as
	\begin{equation}
	\hat{x}[c]=\argmax_{a_j\in \mathbb{A}}\bar p_c({a_j}),\ c=1,\dots,NM. \label{12}
	\end{equation}

Even though ICMP receiver can exploit the SIMO channel
diversity gain, the performance may degenerate
when the corresponding jointly factor graph of the parallel correlated channels is densely connected
to form short cycles \cite{kschischang2001factor}.
Unfortunately, OTFS often exhibits a high
density graph due to off-grid channel delays and Doppler shifts.
In addition, the performance can also suffer when Gaussian approximation of the interference terms becomes less accurate.
To overcome these shortcomings, we propose a more efficient TMP receiver in the next subsection.

\subsection{TMP Receiver Equalization}
Since there are two receive channels from (\ref{FS_V}),
an MP equalizer can be applied for each reception typically.
Given two MP equalizers,
we propose a turbo receiver to enable cooperation
between these two MP equalizers for better performance.
In the TMP receiver, two individual
MP equalizers exchange information in the form of LLRs
for each symbol. The extrinsic LLRs generated by one MP
equalizer are treated as a priori information
by the other. As soft information
is circulated via this algorithmic loop, more reliable
soft information produced in one equalizer helps the other
to improve.
Improved bit error rate (BER) performance can
be achieved as iterations continues and are
terminated
after a certain number $n_{t}$ of iterations.

\begin{figure}
  \centering
  \includegraphics[width=2.5in]{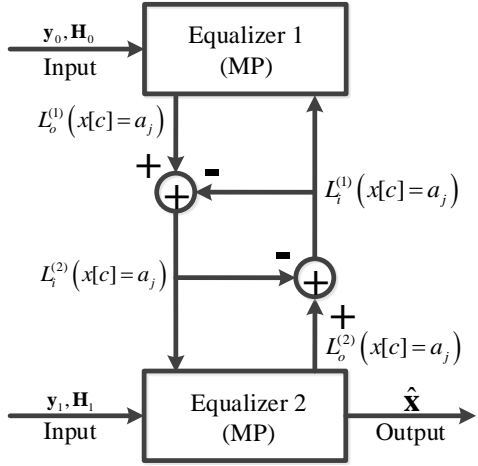}
  \caption{TMP receiver structure.}\label{TMP}
\end{figure}
The TMP receiver structure is shown in Fig. \ref{TMP}.
A similar MP algorithm as in \textbf{Algorithm \ref{alg:A}}
can be employed for each equalizer with only modest modifications of input ${{\bm{\omega }}_c}$ by a priori information and output the a posteriori information.

Specifically, the first equalizer produces output a posteriori LLR on each symbol as
\begin{align}
L_o^{(1)}\left( {x[c] = {a_j}} \right)
= \underbrace {\log \frac{{\Pr ({\bf{y}}|x[c] = {a_j},{\bf{H}})}}{{\Pr ({\bf{y}}|x[c] = {a_Q},{\bf{H}})}}}_{L_e^{(1)}\left( {x[c] = {a_j}} \right)} + \underbrace {\log \frac{{\Pr (x[c] = {a_j})}}{{\Pr (x[c] = {a_Q})}}}_{L_i^{(1)}\left( {x[c] = {a_j}} \right)},
\end{align}
where $c = 1, \cdots ,NM$ and $j = 1, \cdots ,Q$. The extrinsic LLR ${L_e^{(1)}\left( {x[c] = {a_j}} \right)}$ is then passed to the second equalizer as a priori LLR, i.e., $L_i^{(2)}\left( {x[c] = {a_j}} \right) = L_e^{(1)}\left( {x[c] = {a_j}} \right)$. Similarly, the second equalizer generate extrinsic LLR $L_e^{(2)}\left( {x[c] = {a_j}} \right) = L_o^{(2)}\left( {x[c] = {a_j}} \right) - L_i^{(2)}\left( {x[c] = {a_j}} \right)$, which is passed back to the first equalizer as a priori information to form the iterative loop.
Note that we only pass extrinsic information. Otherwise, messages become more and more correlated over the iterations and the efficiency of the iterative algorithm would be reduced, which results in performance loss \cite{douillard1995}.

\subsection{Performance Analysis of TMP Receiver}
Since our proposed FSS receiver can typically generate weakly-correlated channel outputs, the EXIT chart \cite{el2013exit,ten2001convergence} still offers substantially clear insight into the convergence behavior of our proposed TMP receiver. It has been successfully used for analyzing and predicting convergence behavior
of iteratively decoded systems \cite{el2013exit,ten2001convergence}.
In this subsection, we analyze the performance of the proposed TMP receiver by using the tool of EXIT chart \cite{el2013exit,ten2001convergence}, which
tracks the evolution of mutual
information (MI) between transmitted symbols and their
LLRs through iterations. For the EXIT chart of TMP receiver analysis, the two MP equalizers are modeled as the MI transfer devices, i.e., given a priori MI $I_i$ at the input, each equalizer generates
a new extrinsic MI $I_e$ at the output, where ${I_i} \buildrel \Delta \over = I({{\bf{L}}_i};x)$ and ${I_e} \buildrel \Delta \over = I({{\bf{L}}_e};x)$, respectively.

We consider quadrature phase shift keying (QPSK) with Gray mapping as an example, i.e., $\mathbb{A} = \left[ {\frac{{1 + j}}{{\sqrt 2 }},\frac{{1 - j}}{{\sqrt 2 }},\frac{{ - 1 + j}}{{\sqrt 2 }},\frac{{ - 1 - j}}{{\sqrt 2 }}} \right]$ and similar analysis could be done with other modulations. A Gray-mapped QPSK can be regarded as a superposition of the BPSK modulated in-phase and quadrature components, so the MI ${I_i} = I_i^{(\text{I})} + I_i^{(\text{Q})}$ with
\begin{align}\label{MI}
I_i^{(\text{S})} = \frac{1}{2}\sum\limits_{{x^{(\text{S})}} \in \left\{ {{1 \mathord{\left/
 {\vphantom {1 {\sqrt 2 }}} \right.
 \kern-\nulldelimiterspace} {\sqrt 2 }}, - {1 \mathord{\left/
 {\vphantom {1 {\sqrt 2 }}} \right.
 \kern-\nulldelimiterspace} {\sqrt 2 }}} \right\}} {\int_{ - \infty }^\infty  {{f_{L_i^{(\text{S})}}}({l^{(\text{S})}}\left| {{x^{(\text{S})}}} \right.){{\log }_2}\frac{{2{f_{L_i^{(\text{S})}}}({l^{(\text{S})}}\left| {{x^{(\text{S})}}} \right.)}}{{{f_{L_i^{(\text{S})}}}({l^{(\text{S})}}\left| {{1 \mathord{\left/
 {\vphantom {1 {\sqrt 2 }}} \right.
 \kern-\nulldelimiterspace} {\sqrt 2 }}} \right.) + {f_{L_i^{(\text{S})}}}({l^{(\text{S})}}\left| { - {1 \mathord{\left/
 {\vphantom {1 {\sqrt 2 }}} \right.
 \kern-\nulldelimiterspace} {\sqrt 2 }}} \right.)}}d{l^{(\text{S})}}} },
\end{align}
where $\text{S} \in \{ \text{I},\text{Q}\}$ and ${{f_{L_i^{(\text{S})}}}({l^{(\text{S})}}\left| {{x^{(\text{S})}}} \right.)}$ is the conditional distribution of a priori LLR ${L_i^{(\text{S})}}$ given ${x^{(\text{S})}} \in \left\{ {\frac{1}{{\sqrt 2 }}, - \frac{1}{{\sqrt 2 }}} \right\}$. When the Gaussian approximation is applied to ${{f_{L_i^{(\text{S})}}}({l^{(\text{S})}}\left| {{x^{(\text{S})}}} \right.)}$, i.e.,
\begin{align}
{f_{L_i^{(\text{S})}}}({l^{(\text{S})}}\left| {{x^{(\text{S})}}} \right.) = \frac{1}{{\sqrt {2\pi } {\sigma _{L_i^{(\text{S})}}}}}\exp \left( { - \frac{{{{\left({l^{(\text{S})}} - \sqrt 2 \sigma _{L_i^{(\text{S})}}^2{x^{(\text{S})}}\right)}^2}}}{{2\sigma _{L_i^{(\text{S})}}^2}}} \right),
\end{align}
where ${\sigma _{L_i^{(\text{S})}}^2}$ is the variance of the LLR random variables $L_i^{(\text{S})}$. The MI ${I_i}$ can be expressed as
\begin{align}
{I_i}({\sigma _{L_i^{(\text{S})}}})& = 2 - 2\int_{ - \infty }^\infty  {{f_{L_i^{(\text{S})}}}\left({l^{(\text{S})}}\left| {{1 \mathord{\left/
 {\vphantom {1 {\sqrt 2 }}} \right.
 \kern-\nulldelimiterspace} {\sqrt 2 }}} \right.\right){{\log }_2}\left(1 + {e^{ - 2{l^{(\text{S})}}}}\right)d{l^{(\text{S})}}}\nonumber\\
&= 2 - 2{E_{{x^{(\text{S})}} = {1 \mathord{\left/
 {\vphantom {1 {\sqrt 2 }}} \right.
 \kern-\nulldelimiterspace} {\sqrt 2 }}}}\left[{\log _2}\left(1 + {e^{ - 2L_i^{(\text{S})}}}\right)\right],
\end{align}
where ${E_x}[\bullet]$ is the expectation value of $x$. Note that
the function ${I_i}({\sigma _{L_i^{(\text{S})}}})$ is
monotonically increasing and has an inverse.
Additionally, ${\lim _{{\sigma _{L_i^{(\text{S})}}} \to 0}}{I_i}({\sigma _{L_i^{(\text{S})}}}) = 0$ and ${\lim _{{\sigma _{L_i^{(\text{S})}}} \to \infty }}{I_i}({\sigma _{L_i^{(\text{S})}}}) = 2$, which correspond to zero and perfect a priori
information, respectively.

After passing samples of ${{\bf{L}}_i}$ through the MP equalizer,
the output of extrinsic MI $I_e$ is obtained by applying the same expression in (\ref{MI}) with the distribution
of ${{\bf{L}}_e}$.
This can be achieved by first estimating the conditional
distribution of ${{f_{L_e^{(\text{S})}}}({l^{(\text{S})}}\left| {{x^{(\text{S})}}} \right.)}$ using the histogram method \cite{el2013exit,ten2001convergence} before computing  ${I_e} = I_e^{(\text{I})} + I_e^{(\text{Q})}$ numerically based
on (\ref{MI}). The EXIT chart is depicted by repeating the
procedure above for several values of
${\sigma _{L_i^{(\text{S})}}}$ to yield pairs of $({I_i},{I_e})$.

\begin{figure}
\begin{minipage}[t]{0.48\textwidth}
\centering
\includegraphics[width=3.3in]{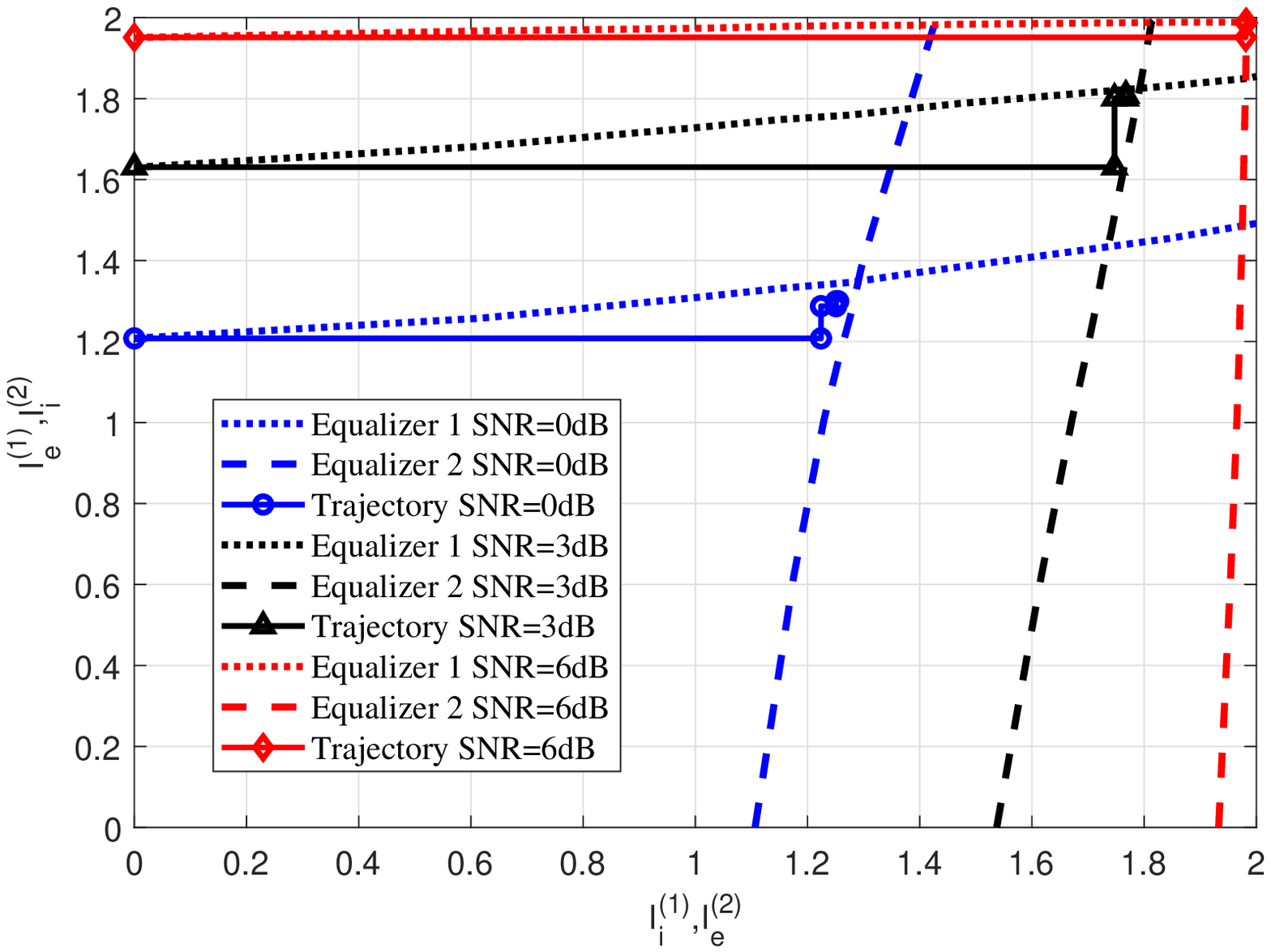}
\caption{EXIT charts and simulated traces of TMP receiver for QPSK with SNR = 0dB, 3dB and 6dB.}\label{EXIT}
\end{minipage}
\hfill
\begin{minipage}[t]{0.48\textwidth}
\centering
\includegraphics[width=3.3in]{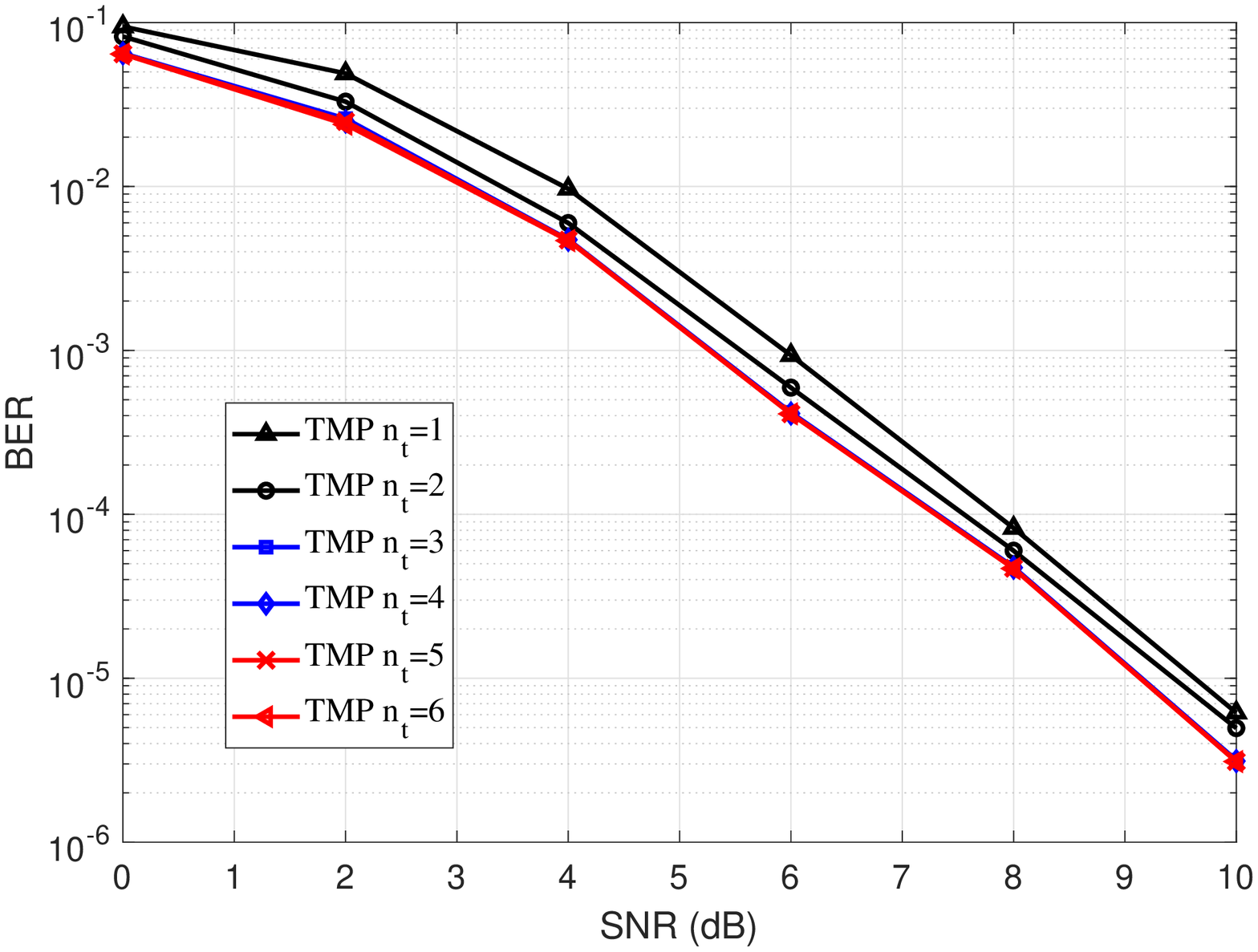}
\caption{BER performance of TMP receiver with different number of iterations for OTFS system.}\label{TMP_ite}
\end{minipage}
\end{figure}
Fig. \ref{EXIT} shows an example of the proposed TMP receiver's
EXIT charts for QPSK at different signal-to-noise ratio (SNR)
levels. Here we select a typical urban channel model \cite{failli1989digital} and generate the Doppler shift for each delay by using the Jakes formulation \cite{khammammetti2018otfs,surabhi2019low,raviteja2018interference} with maximum Doppler frequency shift ${{\nu _{max}}}=1111$ Hz.
We observe that the MI ${I_e}$ increases with ${I_i}$,
which means that the output ${{\bf{L}}_e}$ becomes more
reliable as the input ${{\bf{L}}_i}$ becomes better.
We also show the trajectories of iterative process of the
TMP receiver in Fig. \ref{EXIT}. Note that the system
trajectories closely follow the transfer curves of the two
MP equalizers and eventually reach the corresponding convergence
point (where the transfer curves inter-set) for different SNRs.
The convergence point becomes more reliable as SNR grows,
even approaching the ideal mutual information of 2 bits per
QPSK symbol.

The slight discrepancy between trajectories and
transfer curves can be attributed to the Gaussian model
approximation of the conditional distribution ${{f_{L_i^{(\text{S})}}}({l^{(\text{S})}}\left| {{x^{(\text{S})}}} \right.)}$. In addition, we can
estimate the number of required
iterations for the proposed TMP receiver to converge
by counting the number of staircase steps that follow
the trajectory
curves of Fig. \ref{EXIT}. As we can see,
three iterations are typically sufficient to achieve the desired performance. This analysis is also verified in
Fig. \ref{TMP_ite}, where the performance improvement
becomes negligible beyond three iterations.

%

\section{Reduced Complexity Receivers}\label{S_MP}

From the algorithm discussion, the complexity of the proposed
receivers can be attributed to the MP algorithm.  Clearly, for
each main loop iteration of the MP algorithm, the number of complex multiplications (CMs) required in steps (\ref{5}), (\ref{6}) and (\ref{7}) are $2MNGDQ$, $MNGD(4Q+1)$ and $5MNGDQ$, respectively. Therefore, the overall computational complexity required respectively for ICMP and TMP receivers are $n_{iter}MNGD(11Q+1)$ and $n_t n_{iter}MNGD(11Q+1)$.

We note that the proposed receiver complexity depends critically
on the  number of non-zero channel ISI terms
(i.e., $D$) which represent channel sparsity.
However, $D$ can sometimes remain relatively large, e.g. over
150 in our experiments because of many off-grid delays and
Doppler shifts. Thus, to further reduce receiver complexity,
we propose a simplified MP algorithm by trimming of
some graph edges from participating in message passing
and update. Although such approximation may lead to some
performance loss, edge trimming can also reduce the number of
short cycles in the corresponding factor graph,
which may in fact improve the performance.

The basic idea is to apply Gaussian approximation to part of the channel interferers (i.e., part of the connections),
such that the factor graph can be simplified by trimming these
edges. Specifically, for each observation node $y[d]$,
we would sort the corresponding $D$ channel coefficients
based on their sizes. We choose $R$ largest terms
and the corresponding edges to remain in the graph
while removing the rest.

Through this process, the received signal $y[d]$ in (\ref{4}) can be rewritten as
\begin{align}
y[d] = \sum\limits_{e \in \Phi (d)} H[d,e]x[e] + \underbrace {\sum\limits_{e \in \bar \Phi (d)} H[d,e]x[e] + z[d]}_{z'[d]},
\end{align}
where ${\Phi (d)}$ represents the set of indices with $R$ largest
terms in ${{\cal I}(d)}$ and ${\bar \Phi (d)}$ denotes the set containing the indices for the remaining $(D-R)$ terms.
We use ${z'[d]}$ to denote the
new noise term to be
approximated as a Gaussian random variable with mean and variance:
\begin{align}
{\mu _{z'}}[d] = \sum\limits_{e \in \bar \Phi (d)} {\sum\limits_{j = 1}^Q {{\omega _e}} } ({a_j}){a_j}H[d,e],
\end{align}
\begin{align}
{({\sigma _{z'}}[d])^2} = \sum\limits_{e \in \bar \Phi (d)} {\left( {\sum\limits_{j = 1}^Q {{\omega _e}} ({a_j}){{\left| {{a_j}} \right|}^{\rm{2}}}{{\left| {H[d,e]} \right|}^{\rm{2}}} - {{\left| {\sum\limits_{j = 1}^Q {{\omega _e}} ({a_j}){a_j}H[d,e]} \right|}^{\rm{2}}}} \right)}  + \sigma _N^2.
\end{align}
Therefore, the messages $\mu _{d,c}^{(\kappa)}$ and ${{(\sigma _{d,c}^{(\kappa)})}^2}$ passed from observation node $y[d]$ to variable nodes $x[c],c\in{\Phi (d)}$ in the $\kappa$-th iteration can be expressed as
\begin{align}
\mu _{d,c}^{(\kappa)} = \sum\limits_{e \in \Phi (d),e \ne c} {\sum\limits_{j = 1}^Q {p_{e,d}^{(\kappa - 1)}} } ({a_j}){a_j}H[d,e] + {\mu _{z'}}[d],
\end{align}
\begin{align}
{(\sigma _{d,c}^{(\kappa)})^2} = \sum\limits_{e \in \Phi (d),e \ne c} {\left( {\sum\limits_{j = 1}^Q {p_{e,d}^{(\kappa - 1)}} ({a_j}){{\left| {{a_j}} \right|}^{\rm{2}}}{{\left| {H[d,e]} \right|}^{\rm{2}}} - {{\left| {\sum\limits_{j = 1}^Q {p_{e,d}^{(\kappa - 1)}} ({a_j}){a_j}H[d,e]} \right|}^{\rm{2}}}} \right)}  + {({\sigma _{z'}}[d])^2}.
\end{align}

Similarly, the message passed from variable node $x[c]$ to observation nodes $y[d],d\in \Psi (c)$ in the $\kappa$-th iteration can still be given in (\ref{7}) with only modification of $\alpha _{c,d}^{(\kappa)}({a_j})$ as
\begin{align}
\alpha _{c,d}^{(\kappa)}({a_j}) = \log \frac{{{\omega _c}({a_j})}}{{{\omega _c}({a_Q})}} + \sum\limits_{e \in \Psi (c),e \ne d} {\log \frac{{\varepsilon _{e,c}^{(\kappa)}({a_j})}}{{\varepsilon _{e,c}^{(\kappa)}({a_Q})}}},
\end{align}
where $\Psi (c)$ includes the indices of all the observation nodes that are connected to the variable node $x[c]$ in the simplified factor graph.

As we can see, all the edges participate in message updates
in the original MP algorithm whereas the proposed algorithm
of simplified MP only retains a subset of edges.
Consequently, the overall complexity is reduced to
$n_{iter}MNGR(11Q+1)$ and $n_t n_{iter}MNGR(11Q+1)$ for simplified ICMP (S-ICMP) receiver and simplified TMP (S-TMP) receiver, respectively.

\section{Simulation Results}\label{simulation}

In this section, we test the performance of
our
proposed FSS receivers for OTFS systems in high-mobility time-varying channels.
For simplicity, we consider that the carrier frequency
is $4$ GHz with typical subcarrier spacing ${\Delta f}=15$ kHz. Unless otherwise stated,
Gray-mapped QPSK is the modulation and the RRC rolloff factor
in transmitter and receiver is set to $0.4$.
In addition, we consider $N=32$ time slots and $M=128$
subcarriers in the time-frequency domain.
The speed of the mobile user is set to $\lambda  = 300$ km/h, leading to a maximum Doppler frequency shift ${{\nu _{max}}}=1111$ Hz.
We adopt a typical urban channel model \cite{failli1989digital} with exponentially decaying power delay profile $p(\tau ) = {e^{ - \tau }}$ ($\tau $ is in $\mu$s)
and generate the Doppler shift for each delay by using the Jakes
formulation \cite{khammammetti2018otfs,surabhi2019low,raviteja2018interference},
i.e., ${\nu _i} = {\nu _{max}}\cos ({\rho _i})$, where ${\rho _i}$ is uniformly distributed over $[ - \pi ,\pi ]$.

We first assume that the CSI is known at the receiver. We then
investigate the effect of imperfect CSI on OTFS performance.
Without loss of generality, we choose $\Delta=0.7$, $\varrho=0.1$ and $n_{iter}=20$ for \textbf{Algorithm \ref{alg:A}} and set $G=2$, $n_{t}=3$. All simulation results are from averaging results
over 500 realizations.
\begin{figure}
  \centering
  \includegraphics[width=3.8in]{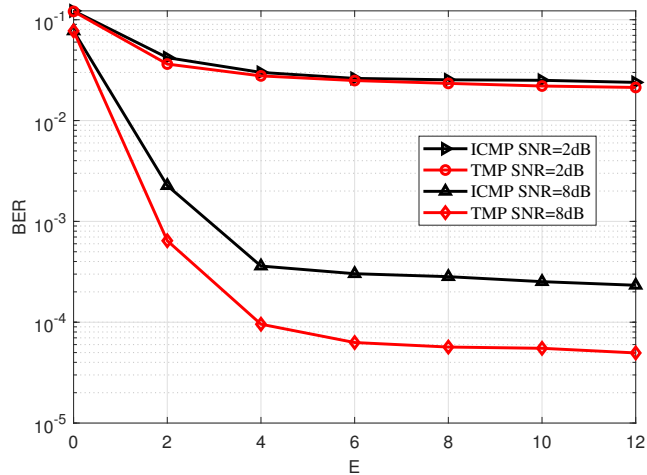}
  \caption{BER performance of OTFS for different numbers of $E$.}\label{E_MP}
\end{figure}

We first study the effects of approximation $E_i$
on OTFS performance. For simplicity, we consider the same
$E_i$ for all paths, i.e., ${E_i} = E,\forall i$.  Fig. \ref{E_MP} illustrates the
BER performance of OTFS system versus different numbers
of $E$ for different receivers under various levels of SNR
(signal quality).
We can see significant performance improvement when
$E$ increases from $0$ to $6$ at the expense of higher complexity.
We also notice a performance saturation thereafter for both
receivers, indicating that, because of many small ISI channel
taps for off-grid delays and Doppler shifts,
very large choices of $E$ do not noticeably
improve receiver performance.
In the rest of our experiments, we shall use $E=6$ unless
otherwise noted.

Fig. \ref{MP_Receiver} compares the BER performance of OTFS system for different receiver designs. To highlight the superiority of the proposed FSS architecture, we also provide the benchmark
performance of traditional SSS receiver
by limiting on-the-grid delay/Doppler shifts in Fig. \ref{MP_Receiver}. The results reveal that every receiver benefits from higher SNR.
However, our proposed FSS receivers outperform SSS receivers significantly owing to the utilization of
channel diversity gain through fractionally spaced sampling.
We also note that the modest BER performance difference
between on-the-grid and off-grid delay/Doppler shifts.
This strongly support the robustness and the practicality of
our proposed receivers given their ability to tackle
any values of delay and Doppler shift.

In general, our proposed TMP receiver achieves
superior performance to ICMP receiver through turbo iterations. This performance advantage stems from the fact that
ICMP receiver suffers from a large number of short cycles
in the channel factor graph and is more prone to
convergence to local optimum.
\begin{figure}
  \centering
  \includegraphics[width=3.8in]{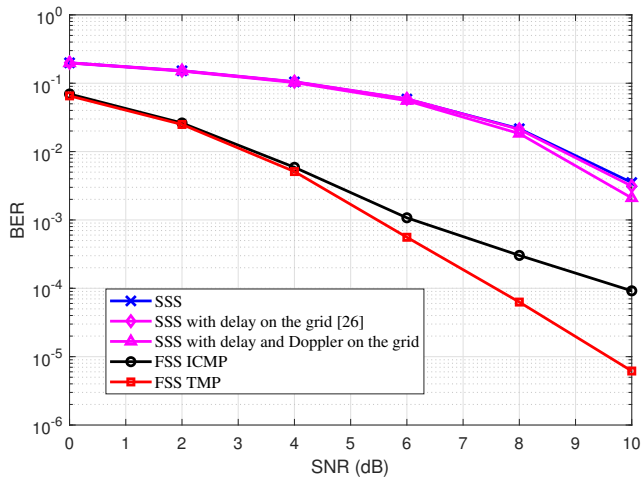}
  \caption{BER performance comparison of OTFS with different receiver designs.}\label{MP_Receiver}
\end{figure}

\begin{figure}
  \centering
  \includegraphics[width=3.8in]{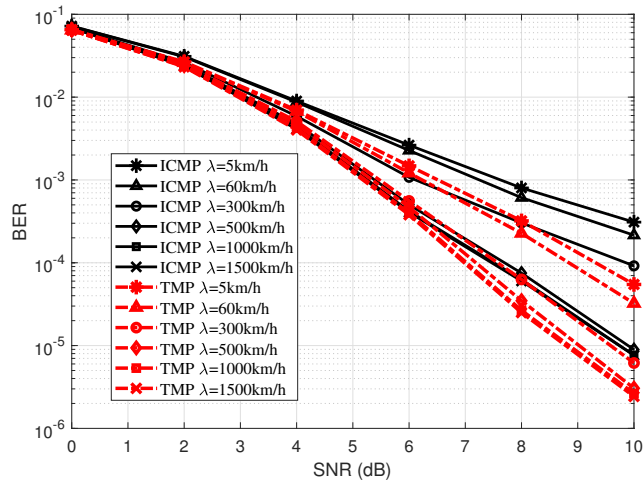}
  \caption{BER performance of OTFS with different user mobile velocities.}\label{MP_Speed}
\end{figure}
Fig. \ref{MP_Speed} shows the BER performance of OTFS system
under various user mobile velocities (i.e., various maximum Doppler shifts).
The results show that the performance improves gradually
as the user velocity increases from 5 km/h to 500 km/h and saturates beyond 500 km/h.
This result would have been surprising to traditional
modulation schemes and equalizers that require quasi-static
channels.
In OTFS, however, the modulation in the delay-Doppler domain
in fact can benefit from larger Doppler shift as a larger
number of multiple paths becomes more distinct.
Our OTFS receiver can resolve a larger number of paths
in the Doppler dimension with the help of higher user velocity.
As a result, better diversity gain becomes possible.

We again notice that the proposed TMP receiver outperforms ICMP receiver for different velocities, which further
exhibits the advantage of TMP receiver over the ICMP receiver
for high mobility users.

Fig. \ref{MP_Grid} shows the BER performance of OTFS transmission with different system parameters. We can observe that the performance of ICMP and TMP receivers degrades as $M$ and $N$ decrease due to the lower resolution of delay-Doppler grid. This leads to the diversity loss since the receiver resolves a smaller number of paths in the channel. We also notice that our FSS receiver can exhibit a certain level of gains even for the high order of modulation (e.g., 16QAM).
These analyses strongly support the consistency of our proposed receivers across different system parameters.
\begin{figure}
  \centering
  \includegraphics[width=3.8in]{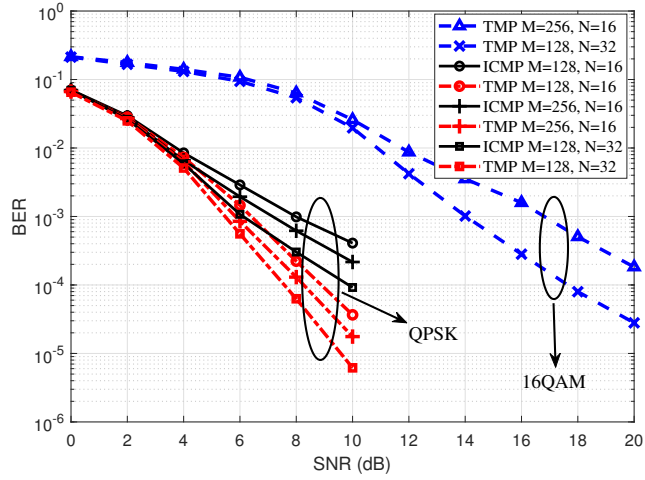}
  \caption{BER performance of OTFS with different system parameters.}\label{MP_Grid}
\end{figure}

For low complexity,
Fig. \ref{SMP} shows the BER performance of OTFS system with the proposed simplified MP receivers.
The results clearly show that as $R$ increases,
the performance of S-ICMP receiver and S-TMP receiver
would approach the performance  of ICMP receiver and TMP receiver,
respectively.
It is worth noting that even with $R=50$,
we already achieve a complexity reduction by the factor
of around 3 since $D$ is around 150 in our simulation. We further note that the
performance loss for the simplified receivers
are rather insignificant even if we select
smaller $R$. Therefore, our proposed simplified MP receivers can provide the desirable trade-off between complexity and performance.

\begin{figure}
\begin{minipage}[t]{0.48\textwidth}
\centering
\includegraphics[width=3.3in]{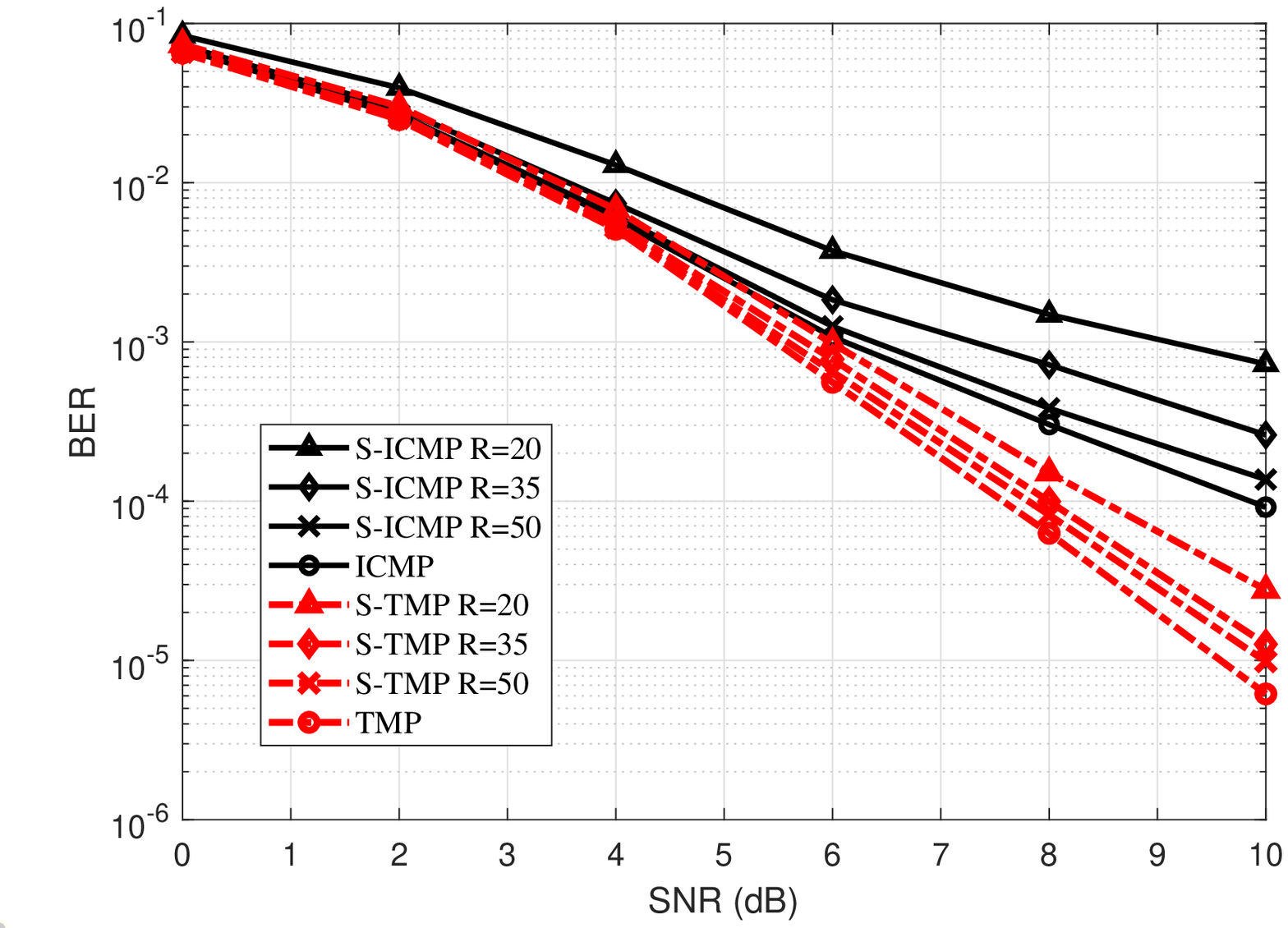}
\caption{BER performance of OTFS with simplified MP receivers.}\label{SMP}
\end{minipage}
\hfill
\begin{minipage}[t]{0.48\textwidth}
\centering
\includegraphics[width=3.3in]{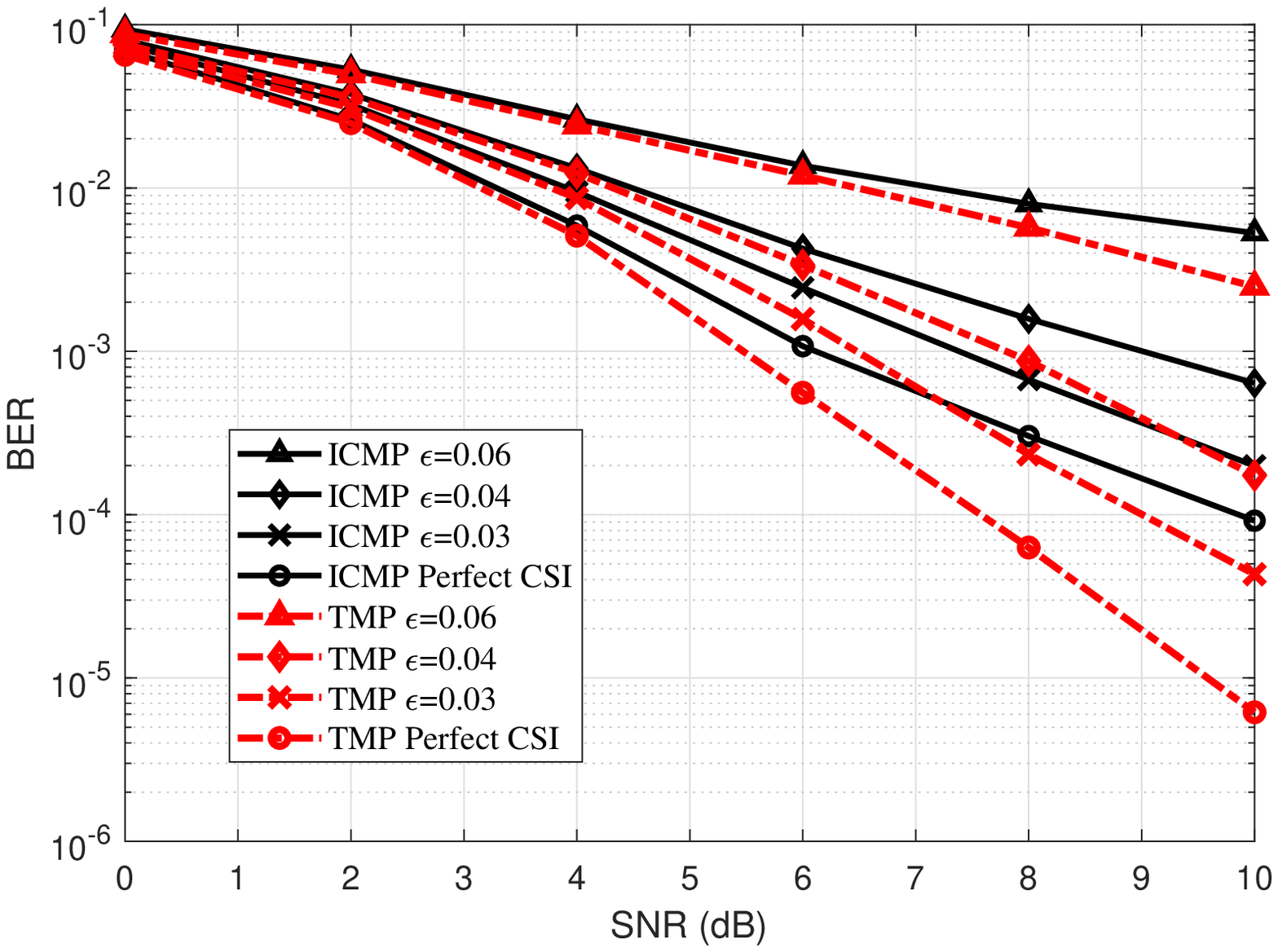}
\caption{BER performance of OTFS with imperfect CSI.}\label{ICSI}
\end{minipage}
\end{figure}
Finally, we test the effect of CSI uncertainty
on the BER performance of OTFS system in Fig. \ref{ICSI}.
In practice, the receiver can only acquire CSI based on
pilots and training which consume power and spectrum
resources.
It is therefore common that receivers must function under
CSI uncertainty. We characterize the CSI error by adopting
the following model \cite{sun2016multi}:
\begin{align*}
{h_i} &= {{\hat h}_i} + \Delta {h_i},\ \left\| {\Delta {h_i}} \right\| \le {\epsilon _{{h_i}}},\nonumber\\
{\tau _i} &= {{\hat \tau }_i} + \Delta {\tau _i},\ \left\| {\Delta {\tau _i}} \right\| \le {\epsilon _{{\tau _i}}},\nonumber\\
{\nu _i} &= {{\hat \nu }_i} + \Delta {\nu _i},\ \left\| {\Delta {\nu _i}} \right\| \le {\epsilon _{{\nu _i}}},\nonumber
\end{align*}
where ${{\hat h}_i}$, ${{\hat \tau }_i}$ and ${{\hat \nu }_i}$ are the estimated versions of ${h_i}$, ${\tau _i}$ and ${\nu _i}$. $\Delta {h_i}$, $\Delta {\tau _i}$ and $\Delta {\nu _i}$ represent the corresponding channel estimation errors, whose norms
are bounded with the given radius ${\epsilon _{{h_i}}}$,
${\epsilon _{{\tau _i}}}$ and ${\epsilon _{{\nu _i}}}$,
respectively. For simplicity, we assume that
${\epsilon _{{h_i}}} = \epsilon \left\| {{{\hat h}_i}} \right\|$,
${\epsilon _{{\tau _i}}} = \epsilon \left\| {{{\hat \tau }_i}}
 \right\|$ and ${\epsilon _{{\nu _i}}} = \epsilon \left\|
 {{{\hat \nu }_i}} \right\|,\forall i$. From Fig. \ref{ICSI},
we can observe mild performance loss for modest
levels of channel uncertainty $\epsilon$.
Without sudden and large drop of receiver performance
as channel uncertainty grows, our proposed new receiver
architecture is robust and can handle typical CSI
errors.


\section{Conclusion}\label{Conclusion}
In this paper, we investigated the design of
practical OTFS receivers to address several practical
considerations. First, when the practical non-ideal
rectangular pulses are used in OTFS transmissions,
we derived the OTFS input-output signal relationship
in the delay-Doppler domain.
Utilizing a compact vectorized form, we illustrated a simple
sparse representation of the channel model.
We further recognized that the use of rectangular OTFS pulses
require bandlimiting pulse shaping filter at the
transmitter and matched filter at the receiver.
Using the traditional RRC pulseshaping,
we developed a fractionally spaced sampling (FSS) framework
for receiver design and proposed two effective receivers
for symbol detection in the delay-Doppler domain.
Our FSS receivers can exploit channel diversity gain and our
EXIT chart analysis demonstrate their rapid convergence.
Furthermore, we proposed simplified MP method to further
reduce the complexity for both the proposed receivers.
Our results demonstrated stronger performance over conventional
receivers and robustness against channel uncertainty
and modeling errors.


%

\appendices
\section{}\label{Appe_A}
Define $V(t,f) = \int {g_{rx}^*(t' - t)N(t'){e^{ - j2\pi f(t' - t)}}dt'}$.
Combining (\ref{transmit_s}), (\ref{receive_r}) and (\ref{receive_Y}), we can rewrite ${Y(t,f)}$
\begin{align}
 =& \int {g_{rx}^*(t' - t)\left[ {\sum\limits_{p = 0}^{P - 1} {\sum\limits_{i = 1}^L {{h_i}{e^{j2\pi {\nu _i}\left( {t' - p{T_s}} \right)}}{{\mathop{\rm P}\nolimits} _\text{{rc}}}(p{T_s} - {\tau _i})} s(t' - p{T_s})} } \right]{e^{ - j2\pi f(t' - t)}}dt'} + V(t,f)\\
 =& \sum\limits_{n' = 0}^{N - 1} {\sum\limits_{m' = 0}^{M - 1} {X[n',m']\sum\limits_{p = 0}^{P - 1} {\sum\limits_{i = 1}^L {{h_i}{{\mathop{\rm P}\nolimits} _\text{{rc}}}(p{T_s} - {\tau _i})} \left[ {\int {g_{rx}^*(t' - t){g_{tx}}(t' - p{T_s} - n'T){e^{j2\pi {\nu _i}\left( {t' - p{T_s}} \right)}}} } \right.} } }\nonumber\\
&\; \left. { \times {e^{j2\pi \left( {m'{\rm{ - }}\frac{{M{\rm{ - 1}}}}{{\rm{2}}}} \right)\Delta f(t' - p{T_s} - n'T)}}{e^{ - j2\pi f(t' - t)}}dt'} \right] + V(t,f).
\end{align}

After sampling, we have
\begin{align}
Y[n,m] = \sum\limits_{n' = 0}^{N - 1} {\sum\limits_{m' = 0}^{M - 1} {{H_{n,m}}[n',m']X[n',m']} }  + V[n,m],
\end{align}
where $V[n,m] = \int {g_{rx}^*(t' - nT)N(t'){e^{ - j2\pi \left( {m{\rm{ - }}\frac{{M{\rm{ - 1}}}}{{\rm{2}}}} \right)\Delta f(t' - nT)}}dt'}$ and
\begin{align}\label{H_TF_G}
{H_{n,m}}[n',m'] &= \sum\limits_{p = 0}^{P - 1} {\sum\limits_{i = 1}^L {{h_i}{{\mathop{\rm P}\nolimits} _\text{{rc}}}(p{T_s} - {\tau _i})\left[ {\int {g_{rx}^*(t' - nT){g_{tx}}(t' - p{T_s} - n'T){e^{j2\pi {\nu _i}\left( {t' - p{T_s}} \right)}}} } \right.} }\nonumber\\
&\quad \left. { \times {e^{j2\pi \left( {m'{\rm{ - }}\frac{{M{\rm{ - 1}}}}{{\rm{2}}}} \right)\Delta f(t' - p{T_s} - n'T)}}{e^{ - j2\pi \left( {m{\rm{ - }}\frac{{M{\rm{ - 1}}}}{{\rm{2}}}} \right)\Delta f(t' - nT)}}dt'} \right].
\end{align}
Changing variable $t'' = t' - p{T_s} - n'T$, we complete
the proof by rewriting (\ref{H_TF_G}) as
\begin{align}
{H_{n,m}}[n',m'] &= \sum\limits_{p = 0}^{P - 1} {\sum\limits_{i = 1}^L {{h_i}{{\mathop{\rm P}\nolimits} _\text{{rc}}}(p{T_s} - {\tau _i})\left[ {\int {g_{rx}^*(t'' - (n - n')T + p{T_s}){g_{tx}}(t''){e^{j2\pi {\nu _i}\left( {t'' + n'T} \right)}}} } \right.} } \nonumber\\
&\quad \left. { \times {e^{j2\pi \left( {m'{\rm{ - }}\frac{{M{\rm{ - 1}}}}{{\rm{2}}}} \right)\Delta ft''}}{e^{ - j2\pi \left( {m{\rm{ - }}\frac{{M{\rm{ - 1}}}}{{\rm{2}}}} \right)\Delta f(t'' - (n - n')T + p{T_s})}}dt''} \right]\\
& = \sum\limits_{p = 0}^{P - 1} {\sum\limits_{i = 1}^L {{h_i}{{\mathop{\rm P}\nolimits} _\text{{rc}}}(p{T_s} - {\tau _i}){A_{{g_{rx}},{g_{tx}}}}\left( {(n - n')T - p{T_s},(m - m')\Delta f - {\nu _i}} \right)} }\nonumber\\
&\quad  \times {e^{j\pi \left( {M - 1} \right)\Delta f(p{T_s} - (n - n')T)}}{e^{j2\pi m'\Delta f((n - n')T - p{T_s})}}{e^{j2\pi {\nu _i}\left( {nT - p{T_s}} \right)}}.
\end{align}

\section{}\label{Appe_B}
Combining (\ref{transmit_X}), (\ref{relat_SFFT}) and (\ref{relation_TF_G}), we have
\begin{align}
 &\frac{1}{{\sqrt {NM} }}\sum\limits_{n = 0}^{N - 1} {\sum\limits_{m = 0}^{M - 1} {\left[ {\sum\limits_{n' = 0}^{N - 1} {\sum\limits_{m' = 0}^{M - 1} {X[n',m']{H_{n,m}}[n',m']} } } \right]{e^{ - j2\pi \left( {\frac{{nk}}{N} - \frac{{m\ell}}{M}} \right)}}} } \\
 = & \frac{1}{{NM}}\sum\limits_{n = 0}^{N - 1} {\sum\limits_{m = 0}^{M - 1} {\left\{ {\sum\limits_{n' = 0}^{N - 1} {\sum\limits_{m' = 0}^{M - 1} {\left[ {\sum\limits_{k' = 0}^{N - 1} {\sum\limits_{\ell' = 0}^{M - 1} {x[k',\ell']{e^{j2\pi \left( {\frac{{n'k'}}{N} - \frac{{m'\ell'}}{M}} \right)}}} } } \right]{H_{n,m}}[n',m']} } } \right\}{e^{ - j2\pi \left( {\frac{{nk}}{N} - \frac{{m\ell}}{M}} \right)}}} } \\ 
=& \frac{1}{{NM}}\sum\limits_{k' = 0}^{N - 1} {\sum\limits_{\ell' = 0}^{M - 1} {{h_{k,\ell}}[k',\ell']x[k',\ell']} }  
\end{align}
by defining ${h_{k,\ell}}[k',\ell'] = \sum\limits_{n = 0}^{N - 1} {\sum\limits_{m = 0}^{M - 1} {\sum\limits_{n' = 0}^{N - 1} {\sum\limits_{m' = 0}^{M - 1} {{H_{n,m}}[n',m']} } {e^{ - j2\pi \left( {\frac{{nk}}{N} - \frac{{m\ell}}{M}} \right)}}{e^{j2\pi \left( {\frac{{n'k'}}{N} - \frac{{m'\ell'}}{M}} \right)}}} }$.
Hence, we can write $ y[k,\ell] = \frac{1}{{NM}}\sum\limits_{k' = 0}^{N - 1} {\sum\limits_{\ell' = 0}^{M - 1} {{h_{k,\ell}}[k',\ell']x[k',\ell']} } +
\upsilon [k,\ell]$, which completes the proof.

\section{}\label{Appe_C}
Combining (\ref{relation_DD_G}), (\ref{rela_h}) and (\ref{relat_H_R_CP}), we derive the OTFS input-output relationship in delay-Doppler domain separately for $n' = n$ and $n' = {\left[ {n - 1} \right]_N}$.

We first define
\begin{align}
{G_c}({\nu _i}) = {G_s}({\nu _i}) = \sum\limits_{n = 0}^{N - 1} {{e^{ - j2\pi n\left( {\frac{{k - k'}}{N} - {\nu _i}T} \right)}}}  = \frac{{{e^{ - j2\pi ( {k - k' - {k_{{\nu _i}}} - {\beta _{{\nu _i}}}} )}} - 1}}{{{e^{ - j\frac{{2\pi }}{N}( {k - k' - {k_{{\nu _i}}} - {\beta _{{\nu _i}}}} )}} - 1}},
\end{align}
\begin{align}
{F_c}({\nu _i})& = \frac{1}{M}\sum\limits_{c = 0}^{M - 1 - p} {{e^{j2\pi {\nu _i}\left( {\frac{c}{{M\Delta f}} + p{T_s}} \right)}}\sum\limits_{m = 0}^{M - 1} {{e^{ - j2\pi m\left( {\frac{c}{M} + \Delta fp{T_s} - \frac{\ell}{M}} \right)}}} \sum\limits_{m' = 0}^{M - 1} {{e^{j2\pi m'\left( {\frac{c}{M} - \frac{{\ell'}}{M}} \right)}}} }\nonumber \\
& = M\sum\limits_{c = 0}^{M - 1 - p} {{e^{j2\pi \frac{{{k_{{\nu _i}}} + {\beta _{{\nu _i}}}}}{{NM}}\left( {c + p} \right)}}\delta \left( {{{\left[ {c + p - \ell} \right]}_M}} \right)\delta \left( {{{\left[ {c - \ell'} \right]}_M}} \right)},
\end{align}
and
\begin{align}
{F_s}({\nu _i}) &= \frac{1}{M}\sum\limits_{s = M - p}^{M - 1} {{e^{j2\pi {\nu _i}\left( {\frac{s}{{M\Delta f}} + p{T_s} - T} \right)}}\sum\limits_{m = 0}^{M - 1} {{e^{ - j2\pi m\left( {\frac{s}{M} + \Delta fp{T_s} - \Delta fT - \frac{\ell}{M}} \right)}}} \sum\limits_{m' = 0}^{M - 1} {{e^{j2\pi m'\left( {\frac{s}{M} - \frac{{\ell'}}{M}} \right)}}} }\nonumber\\
& = M\sum\limits_{s = M - p}^{M - 1} {{e^{j2\pi \frac{{{k_{{\nu _i}}} + {\beta _{{\nu _i}}}}}{{NM}}\left( {s + p - M} \right)}}\delta \left( {{{\left[ {s + p - \ell} \right]}_M}} \right)\delta \left( {{{\left[ {s - \ell'} \right]}_M}} \right)}.
\end{align}
We also denote
\begin{align}
{\gamma _c}(\ell,p,q,{k_{{\nu _i}}},{\beta _{{\nu _i}}}) = \frac{1}{N}\xi (\ell,p,{k_{{\nu _i}}},{\beta _{{\nu _i}}})\theta (q,{\beta _{{\nu _i}}}),
\end{align}
\begin{align}
{\gamma _s}(k,\ell,p,q,{k_{{\nu _i}}},{\beta _{{\nu _i}}}) = \frac{1}{N}\xi (\ell,p,{k_{{\nu _i}}},{\beta _{{\nu _i}}})\theta (q,{\beta _{{\nu _i}}})\phi (k,q,{k_{{\nu _i}}}),
\end{align}
where $\xi (\ell,p,{k_{{\nu _i}}},{\beta _{{\nu _i}}})$, $\theta (q,{\beta _{{\nu _i}}})$ and $\phi (k,q,{k_{{\nu _i}}})$ are defined in (\ref{Phase_off}), (\ref{Theat_off}) and (\ref{Phi_off}), respectively.

When $n' = n$, we have
\begin{align}
{y_c}[k,\ell] = \frac{1}{{NM}}\sum\limits_{k' = 0}^{N - 1} {\sum\limits_{\ell' = 0}^{M - 1} {h_{k,\ell}^c[k',\ell']x[k',\ell']} },
\end{align}
where
\begin{align}
h_{k,\ell}^c[k',\ell'] &= \sum\limits_{n = 0}^{N - 1} {\sum\limits_{m = 0}^{M - 1} {\sum\limits_{m' = 0}^{M - 1} {{H_{n,m}}[n,m']{e^{ - j2\pi n\left( {\frac{{k - k'}}{N}} \right)}}{e^{j2\pi \left( {\frac{{m\ell - m'\ell'}}{M}} \right)}}} } }\label{h_c}\\
 & = \sum\limits_{n = 0}^{N - 1} {\sum\limits_{m = 0}^{M - 1} {\sum\limits_{m' = 0}^{M - 1} {\left[ {\sum\limits_{p = 0}^{P - 1} {\sum\limits_{i = 1}^L {{h_i}{{\mathop{\rm P}\nolimits} _\text{{rc}}}(p{T_s} - {\tau _i}){A_{{g_{rx}},{g_{tx}}}}\left( { - p{T_s},(m - m')\Delta f - {\nu _i}} \right)} } } \right.} } }\nonumber\\
&\quad  \left. {\times{e^{j\pi \left( {M - 1} \right)\Delta fp{T_s}}}{e^{ - j2\pi m'\Delta fp{T_s}}}{e^{j2\pi {\nu _i}\left( {nT - p{T_s}} \right)}}} \right]{e^{ - j2\pi n\left( {\frac{{k - k'}}{N}} \right)}}{e^{j2\pi \left( {\frac{{m\ell - m'\ell'}}{M}} \right)}}\label{h_c_R}\\
& = \sum\limits_{n = 0}^{N - 1} {\sum\limits_{m = 0}^{M - 1} {\sum\limits_{m' = 0}^{M - 1} {\left[ {\sum\limits_{p = 0}^{P - 1} {\sum\limits_{i = 1}^L {{h_i}{{\mathop{\rm P}\nolimits} _\text{{rc}}}(p{T_s} - {\tau _i})\frac{1}{M}\sum\limits_{c = 0}^{M - 1 - p} {{e^{ - j2\pi \left( {(m - m')\Delta f - {\nu _i}} \right)\left( {\frac{c}{{M\Delta f}} + p{T_s}} \right)}}} } } } \right.} } } \nonumber\\
&\quad  \left. {\times{e^{j\pi \left( {M - 1} \right)\Delta fp{T_s}}}{e^{ - j2\pi m'\Delta fp{T_s}}}{e^{j2\pi {\nu _i}\left( {nT - p{T_s}} \right)}}} \right]{e^{ - j2\pi n\left( {\frac{{k - k'}}{N}} \right)}}{e^{j2\pi \left( {\frac{{m\ell - m'\ell'}}{M}} \right)}}\label{h_c_R_D}\\
& = \sum\limits_{p = 0}^{P - 1} {\sum\limits_{i = 1}^L {{h_i}{{\mathop{\rm P}\nolimits} _\text{{rc}}}(p{T_s} - {\tau _i}){e^{j\pi \frac{{M - 1}}{M}p}}{e^{ - j2\pi {\nu _i}p{T_s}}}{G_c}({\nu _i}){F_c}({\nu _i})} }\label{h_c_R_D_L},
\end{align}

By substituting (\ref{relat_H_R_CP}) in (\ref{h_c}), $h_{k,\ell}^c[k',\ell']$ can be written as in (\ref{h_c_R}), which can be further
written as in (\ref{h_c_R_D}) by replacing the cross-ambiguity
function ${{A_{{g_{rx}},{g_{tx}}}}\left( { - p{T_s},(m - m')\Delta f - {\nu _i}} \right)}$ in (\ref{h_c_R}) with its sampled version. Finally, we obtain $h_{k,\ell}^c[k',\ell']$ in (\ref{h_c_R_D_L})
by separating the terms related to $n$, $m$, $m'$ and $c$, respectively. As a result, we have
\begin{align}
&{y_c}[k,\ell] = \frac{1}{N}\sum\limits_{p = 0}^{P - 1} {\sum\limits_{i = 1}^L {{h_i}{{\mathop{\rm P}\nolimits} _\text{{rc}}}(p{T_s} - {\tau _i}){e^{j\pi \frac{{M - 1}}{M}p}}{e^{ - j2\pi {\nu _i}p{T_s}}}} } \left[ {\sum\limits_{\ell' = 0}^{M - 1} {\sum\limits_{c = 0}^{M - 1 - p} {{e^{j2\pi \frac{{{k_{{\nu _i}}} + {\beta _{{\nu _i}}}}}{{NM}}\left( {c + p} \right)}}} } } \right.\nonumber\\
&\qquad\qquad  \left. {\times\delta \left( {{{\left[ {c + p - \ell} \right]}_M}} \right)\delta \left( {{{\left[ {c - \ell'} \right]}_M}} \right)\sum\limits_{k' = 0}^{N - 1} {{G_c}({\nu _i})x[k',\ell']} } \right]\\
&=\begin{cases}
\sum\limits_{p = 0}^{P - 1} {\sum\limits_{i = 1}^L {\sum\limits_{q = 0}^{N - 1} {{h_i}{{\mathop{\rm P}\nolimits} _\text{{rc}}}(p{T_s} - {\tau _i}){\gamma _c}(\ell,p,q,{k_{{\nu _i}}},{\beta _{{\nu _i}}})x\left[ {{{\left[ {k - {k_{{\nu _i}}} + q} \right]}_N},{{\left[ {\ell - p} \right]}_M}} \right]} } },&p \le \ell < M,\\
0, &\text{otherwise},
\end{cases}\label{relat_DD_c}
\end{align}
where the last equality follows from the change of variable ${k'{\rm{ = }}{{\left[ {k - {k_{{\nu _i}}} + q} \right]}_N}}$.

In a similar fashion, for $n' = {\left[ {n - 1} \right]_N}$, we have
\begin{align}
{y_s}[k,\ell] = \frac{1}{{NM}}\sum\limits_{k' = 0}^{N - 1} {\sum\limits_{\ell' = 0}^{M - 1} {{e^{ - j2\pi \frac{{k'}}{N}}}h_{k,\ell}^s[k',\ell']x[k',\ell']} },
\end{align}
where
\begin{align}
h_{k,\ell}^s[k',\ell'] &= \sum\limits_{n = 0}^{N - 1} {\sum\limits_{m = 0}^{M - 1} {\sum\limits_{m' = 0}^{M - 1} {{H_{n,m}}\left[ {{{\left[ {n - 1} \right]}_N},m'} \right]{e^{ - j2\pi n\left( {\frac{{k - k'}}{N}} \right)}}{e^{j2\pi \left( {\frac{{m\ell - m'\ell'}}{M}} \right)}}} } }\\
& = \sum\limits_{n = 0}^{N - 1} {\sum\limits_{m = 0}^{M - 1} {\sum\limits_{m' = 0}^{M - 1} {\left[ {\sum\limits_{p = 0}^{P - 1} {\sum\limits_{i = 1}^L {{h_i}{{\mathop{\rm P}\nolimits} _\text{{rc}}}(p{T_s} - {\tau _i})\frac{1}{M}\sum\limits_{s = M - p}^{M - 1} {{e^{ - j2\pi \left( {(m - m')\Delta f - {\nu _i}} \right)\left( {\frac{s}{{M\Delta f}} + p{T_s} - T} \right)}}} } } } \right.} } } \nonumber\\
&\quad\left. { \times {e^{j\pi \left( {M - 1} \right)\Delta f(p{T_s} - T)}}{e^{j2\pi m'\Delta f(T - p{T_s})}}{e^{j2\pi {\nu _i}\left( {nT - p{T_s}} \right)}}} \right]{e^{ - j2\pi n\left( {\frac{{k - k'}}{N}} \right)}}{e^{j2\pi \left( {\frac{{m\ell - m'\ell'}}{M}} \right)}}\\
& = \sum\limits_{p = 0}^{P - 1} {\sum\limits_{i = 1}^L {{h_i}{{\mathop{\rm P}\nolimits} _\text{{rc}}}(p{T_s} - {\tau _i}){e^{j\pi \frac{{M - 1}}{M}p}}{e^{ - j\pi \left( {M - 1} \right)}}{e^{ - j2\pi {\nu _i}p{T_s}}}{G_s}({\nu _i}){F_s}({\nu _i})} }.
\end{align}

Thus, ${y_s}[k,\ell]$ can be obtained as
\begin{align}
&{y_s}[k,\ell] = \frac{1}{N}\sum\limits_{p = 0}^{P - 1} {\sum\limits_{i = 1}^L {{h_i}{{\mathop{\rm P}\nolimits} _\text{{rc}}}(p{T_s} - {\tau _i}){e^{j\pi \frac{{M - 1}}{M}p}}{e^{ - j2\pi {\nu _i}p{T_s}}}{e^{ - j\pi \left( {M - 1} \right)}}\left[ {\sum\limits_{\ell' = 0}^{M - 1} {\sum\limits_{s = M - p}^{M - 1} {{e^{j2\pi \frac{{{k_{{\nu _i}}} + {\beta _{{\nu _i}}}}}{{NM}}\left( {s + p - M} \right)}}} } } \right.} }\nonumber\\
& \qquad\qquad  \left. {\times\delta \left( {{{\left[ {s + p - \ell} \right]}_M}} \right)\delta \left( {{{\left[ {s - \ell'} \right]}_M}} \right)\sum\limits_{k' = 0}^{N - 1} {{G_s}({\nu _i}){e^{ - j2\pi \frac{{k'}}{N}}}x[k',\ell']} } \right]\\
&=
\begin{cases}
 \sum\limits_{p = 0}^{P - 1} {\sum\limits_{i = 1}^L {\sum\limits_{q = 0}^{N - 1} {{h_i}{{\mathop{\rm P}\nolimits} _\text{{rc}}}(p{T_s} - {\tau _i}){\gamma _s}(k,\ell,p,q,{k_{{\nu _i}}},{\beta _{{\nu _i}}})x\left[ {{{\left[ {k - {k_{{\nu _i}}} + q} \right]}_N},{{\left[ {\ell - p} \right]}_M}} \right]} } },&0 \le \ell < p,\\
0, &\text{otherwise}.
\end{cases}\label{relat_DD_s}
\end{align}

Finally, by combining (\ref{relat_DD_c}) and (\ref{relat_DD_s}), the input-output relationship of OTFS in delay-Doppler domain can be obtained as in (\ref{relation_DD_R}), which completes the proof.

\ifCLASSOPTIONcaptionsoff
  \newpage
\fi



%




\bibliographystyle{IEEEtran}
\footnotesize
\bibliography{ref_OTFS}

\begin{thebibliography}{10}
\providecommand{\url}[1]{#1}
\csname url@samestyle\endcsname
\providecommand{\newblock}{\relax}
\providecommand{\bibinfo}[2]{#2}
\providecommand{\BIBentrySTDinterwordspacing}{\spaceskip=0pt\relax}
\providecommand{\BIBentryALTinterwordstretchfactor}{4}
\providecommand{\BIBentryALTinterwordspacing}{\spaceskip=\fontdimen2\font plus
\BIBentryALTinterwordstretchfactor\fontdimen3\font minus
  \fontdimen4\font\relax}
\providecommand{\BIBforeignlanguage}[2]{{%
\expandafter\ifx\csname l@#1\endcsname\relax
\typeout{** WARNING: IEEEtran.bst: No hyphenation pattern has been}%
\typeout{** loaded for the language `#1'. Using the pattern for}%
\typeout{** the default language instead.}%
\else
\language=\csname l@#1\endcsname
\fi
#2}}
\providecommand{\BIBdecl}{\relax}
\BIBdecl

\bibitem{wang2006performance}
T.~Wang, J.~G. Proakis, E.~Masry, and J.~R. Zeidler, ``Performance degradation
  of {OFDM} systems due to {D}oppler spreading,'' \emph{IEEE Trans. Wireless
  Commun.}, vol.~5, no.~6, pp. 1422--1432, Jun. 2006.

\bibitem{cai2003bounding}
X.~Cai and G.~B. Giannakis, ``Bounding performance and suppressing intercarrier
  interference in wireless mobile {OFDM},'' \emph{IEEE Trans. Commun.},
  vol.~51, no.~12, pp. 2047--2056, Dec. 2003.

\bibitem{das2007max}
S.~Das and P.~Schniter, ``Max-{SINR} {ISI}/{ICI}-shaping multicarrier
  communication over the doubly dispersive channel,'' \emph{IEEE Trans. Signal
  Process.}, vol.~55, no.~12, pp. 5782--5795, Dec. 2007.

\bibitem{zhao2001intercarrier}
Y.~Zhao and S.-G. Haggman, ``Intercarrier interference self-cancellation scheme
  for {OFDM} mobile communication systems,'' \emph{IEEE Trans. Commun.},
  vol.~49, no.~7, pp. 1185--1191, Jul. 2001.

\bibitem{dean2017new}
T.~Dean, M.~Chowdhury, and A.~Goldsmith, ``A new modulation technique for
  {D}oppler compensation in frequency-dispersive channels,'' in \emph{Proc.
  IEEE PIMRC}, Montreal, QC, Canada, Oct. 2017, pp. 1--7.

\bibitem{hadani2017orthogonal}
R.~Hadani, S.~Rakib, M.~Tsatsanis, A.~Monk, A.~J. Goldsmith, A.~F. Molisch, and
  R.~Calderbank, ``Orthogonal time frequency space modulation,'' in \emph{Proc.
  IEEE WCNC}, San Francisco, CA, USA, Mar. 2017, pp. 1--6.

\bibitem{ramachandran2018mimo}
M.~K. Ramachandran and A.~Chockalingam, ``{MIMO}-{OTFS} in high-{D}oppler
  fading channels: Signal detection and channel estimation,'' in \emph{Proc.
  IEEE Global Commun. Conf. (GLOBECOM)}, Abu Dhabi, United Arab Emirates, Dec.
  2018, pp. 206--212.

\bibitem{khammammetti2018otfs}
V.~Khammammetti and S.~K. Mohammed, ``{OTFS}-based multiple-access in high
  {D}oppler and delay spread wireless channels,'' \emph{IEEE Wireless Commun.
  Lett.}, vol.~8, no.~2, pp. 528--531, Apr. 2018.

\bibitem{ding2019otfs}
Z.~Ding, R.~Schober, P.~Fan, and H.~V. Poor, ``{OTFS}-{NOMA}: An efficient
  approach for exploiting heterogenous user mobility profiles,'' \emph{IEEE
  Trans. Commun.}, vol.~67, no.~11, pp. 7950--7965, Nov. 2019.

\bibitem{raviteja2019orthogonal}
P.~Raviteja, K.~T. Phan, Y.~Hong, and E.~Viterbo, ``Orthogonal time frequency
  space ({OTFS}) modulation based radar system,'' in \emph{Proc. IEEE Radar
  Conf. (RadarConf)}, 2019, pp. 1--6.

\bibitem{gaudio2019performance}
L.~Gaudio, M.~Kobayashi, B.~Bissinger, and G.~Caire, ``Performance analysis of
  joint radar and communication using {OFDM} and {OTFS},'' in \emph{Proc. IEEE
  Int. Conf. Commun. Workshops (ICC Workshops)}, May 2019, pp. 1--6.

\bibitem{raviteja2019otfs}
P.~Raviteja, E.~Viterbo, and Y.~Hong, ``{OTFS} performance on static multipath
  channels,'' \emph{IEEE Wireless Commun. Lett.}, vol.~8, no.~3, pp. 745--748,
  Jun. 2019.

\bibitem{surabhi2019diversity}
G.~Surabhi, R.~M. Augustine, and A.~Chockalingam, ``On the diversity of uncoded
  {OTFS} modulation in doubly-dispersive channels,'' \emph{IEEE Trans. Wireless
  Commun.}, vol.~18, no.~6, pp. 3049--3063, Jun. 2019.

\bibitem{raviteja2019effective}
P.~Raviteja, Y.~Hong, E.~Viterbo, and E.~Biglieri, ``Effective diversity of
  {OTFS} modulation,'' \emph{IEEE Wireless Commun. Lett.}, vol.~9, no.~2, pp.
  249--253, Feb. 2020.

\bibitem{surabhi2019peak}
G.~Surabhi, R.~M. Augustine, and A.~Chockalingam, ``Peak-to-average power ratio
  of {OTFS} modulation,'' \emph{IEEE Commun. Lett.}, vol.~23, no.~6, pp.
  999--1002, Jun. 2019.

\bibitem{raviteja2018practical}
P.~Raviteja, Y.~Hong, E.~Viterbo, and E.~Biglieri, ``Practical pulse-shaping
  waveforms for reduced-cyclic-prefix {OTFS},'' \emph{IEEE Trans. Veh. Tech.},
  vol.~68, no.~1, pp. 957--961, Jan. 2019.

\bibitem{surabhi2019low}
G.~Surabhi and A.~Chockalingam, ``Low-complexity linear equalization for {OTFS}
  modulation,'' \emph{IEEE Commun. Lett.}, vol.~24, no.~2, pp. 330--334, Feb.
  2020.

\bibitem{murali2018otfs}
K.~Murali and A.~Chockalingam, ``On {OTFS} modulation for high-{D}oppler fading
  channels,'' in \emph{Proc. Inform. Theory and Applications Workshop (ITA)},
  Feb. 2018, pp. 1--10.

\bibitem{raviteja2019embedded}
P.~Raviteja, K.~T. Phan, and Y.~Hong, ``Embedded pilot-aided channel estimation
  for {OTFS} in delay--{D}oppler channels,'' \emph{IEEE Trans. Veh. Tech.},
  vol.~68, no.~5, pp. 4906--4917, May 2019.

\bibitem{shen2019channel}
W.~Shen, L.~Dai, J.~An, P.~Fan, and R.~W. Heath, ``Channel estimation for
  orthogonal time frequency space ({OTFS}) massive {MIMO},'' \emph{IEEE Trans.
  Signal Process.}, vol.~67, no.~16, pp. 4204--4217, Aug. 2019.

\bibitem{matz2013time}
G.~Matz, H.~Bolcskei, and F.~Hlawatsch, ``Time-frequency foundations of
  communications: Concepts and tools,'' \emph{IEEE Signal Process. Mag.},
  vol.~30, no.~6, pp. 87--96, Nov. 2013.

\bibitem{farhang2017low}
A.~Farhang, A.~RezazadehReyhani, L.~E. Doyle, and B.~Farhang-Boroujeny, ``Low
  complexity modem structure for {OFDM}-based orthogonal time frequency space
  modulation,'' \emph{IEEE Wireless Commun. Lett.}, vol.~7, no.~3, pp.
  344--347, Jun. 2018.

\bibitem{rezazadehreyhani2018analysis}
A.~RezazadehReyhani, A.~Farhang, M.~Ji, R.~R. Chen, and B.~Farhang-Boroujeny,
  ``Analysis of discrete-time {MIMO} {OFDM}-based orthogonal time frequency
  space modulation,'' in \emph{Proc. IEEE Int. Conf. Commun. (ICC)}, May 2018,
  pp. 1--6.

\bibitem{li2019low}
L.~Li, Y.~Liang, P.~Fan, and Y.~Guan, ``Low complexity detection algorithms for
  {OTFS} under rapidly time-varying channel,'' in \emph{Proc. IEEE 89th Veh.
  Tech. Conf. (VTC-Spring)}, Apr. 2019, pp. 1--5.

\bibitem{long2019low}
F.~Long, K.~Niu, C.~Dong, and J.~Lin, ``Low complexity iterative {LMMSE}-{PIC}
  equalizer for {OTFS},'' in \emph{Proc. IEEE Int. Conf. Commun. (ICC)}, May
  2019, pp. 1--6.

\bibitem{raviteja2018interference}
P.~Raviteja, K.~T. Phan, Y.~Hong, and E.~Viterbo, ``Interference cancellation
  and iterative detection for orthogonal time frequency space modulation,''
  \emph{IEEE Trans. Wireless Commun.}, vol.~17, no.~10, pp. 6501--6515, Oct.
  2018.

\bibitem{tiwari2019low}
S.~Tiwari, S.~S. Das, and V.~Rangamgari, ``Low complexity {LMMSE} receiver for
  {OTFS},'' \emph{IEEE Commun. Lett.}, vol.~23, no.~12, pp. 2205--2209, Dec.
  2019.

\bibitem{ding2001blind}
Z.~Ding and Y.~Li, \emph{Blind equalization and identification}.\hskip 1em plus
  0.5em minus 0.4em\relax CRC press, 2001.

\bibitem{tepedelenlioglu2004low}
C.~Tepedelenlioglu and R.~Challagulla, ``Low-complexity multipath diversity
  through fractional sampling in {OFDM},'' \emph{IEEE Trans. Signal Process.},
  vol.~52, no.~11, pp. 3104--3116, Nov. 2004.

\bibitem{wu2010oversampled}
J.~Wu and Y.~R. Zheng, ``Oversampled orthogonal frequency division multiplexing
  in doubly selective fading channels,'' \emph{IEEE Trans. Commun.}, vol.~59,
  no.~3, pp. 815--822, Mar. 2011.

\bibitem{forney1972maximum}
G.~Forney, ``Maximum-likelihood sequence estimation of digital sequences in the
  presence of intersymbol interference,'' \emph{IEEE Trans. Inform. Theory},
  vol.~18, no.~3, pp. 363--378, May 1972.

\bibitem{brosamler1988almost}
G.~A. Brosamler, ``An almost everywhere central limit theorem,'' in \emph{Math.
  Proc. Cambridge Philos. Soc.}, vol. 104, no.~3.\hskip 1em plus 0.5em minus
  0.4em\relax Cambridge University Press, 1988, pp. 561--574.

\bibitem{som2011low}
P.~Som, T.~Datta, N.~Srinidhi, A.~Chockalingam, and B.~S. Rajan,
  ``Low-complexity detection in large-dimension {MIMO}-{ISI} channels using
  graphical models,'' \emph{IEEE J. Sel. Topics Signal Process.}, vol.~5,
  no.~8, pp. 1497--1511, Dec. 2011.

\bibitem{kschischang2001factor}
F.~R. Kschischang, B.~J. Frey, and H.-A. Loeliger, ``Factor graphs and the
  sum-product algorithm,'' \emph{IEEE Trans. Inform. Theory}, vol.~47, no.~2,
  pp. 498--519, Feb. 2001.

\bibitem{douillard1995}
C.~Douillard \emph{et~al.}, ``Iterative correction of intersymbol interference:
  Turbo-equalization,'' \emph{Eur. Trans. Telecommun.}, vol.~6, no.~5, pp.
  507--511, Sep. 1995.

\bibitem{el2013exit}
M.~El-Hajjar and L.~Hanzo, ``{EXIT} charts for system design and analysis,''
  \emph{IEEE Commun. Surveys Tuts.}, vol.~16, no.~1, pp. 127--153, 1st Quart.
  2014.

\bibitem{ten2001convergence}
S.~Ten~Brink, ``Convergence behavior of iteratively decoded parallel
  concatenated codes,'' \emph{IEEE Trans. Commun.}, vol.~49, no.~10, pp.
  1727--1737, Oct. 2001.

\bibitem{failli1989digital}
M.~Failli, \emph{Digital Land Mobile Radio Communications. COST 207}.\hskip 1em
  plus 0.5em minus 0.4em\relax European Communities, Luxembourg, 1989.

\bibitem{sun2016multi}
Y.~Sun, D.~W.~K. Ng, J.~Zhu, and R.~Schober, ``Multi-objective optimization for
  robust power efficient and secure full-duplex wireless communication
  systems,'' \emph{IEEE Trans. Wireless Commun.}, vol.~15, no.~8, pp.
  5511--5526, Aug. 2016.

\end{thebibliography}

%




\end{document}